\documentclass[a4paper,USenglish,cleveref,autoref,thm-restate]{lipics-v2021}
\usepackage{algorithm}
\usepackage[noend]{algpseudocode}

\hideLIPIcs  

\nolinenumbers

\newcommand{\Path}[1]{\left< #1 \right>}
\newcommand{\eps}{\epsilon}

\newcommand{\Dist}{\mathsf{dist}}
\newcommand{\DDist}{\mathsf{ddist}}
\newcommand{\JDist}{\mathsf{jdist}}

\newcommand{\Minlength}{\mathsf{minL}}
\newcommand{\Minpath}{\mathsf{MinP}}
\newcommand{\Iset}{\Psi}

\newcommand{\Query}{\textsf{dso-query}}
\newcommand{\Fquery}{\textsf{pf-query}}
\newcommand{\Dquery}{\textsf{dp-query}}

\newcounter{cntLemmaNumber}
\newenvironment{rlemma}[1]{%
\setcounter{cntLemmaNumber}{\thelemma}
\setcounterref{lemma}{#1}
\addtocounter{lemma}{-1}
\begin{lemma}
}{%
\end{lemma}
\setcounter{lemma}{\thecntLemmaNumber}
}
\newcounter{cntTheoremNumber}

\newcounter{cntPropositionNumber}

\title{A Nearly Linear Time Construction of Approximate Single-Source Distance Sensitivity Oracles} 

\titlerunning{Approximate Single-Source Distance Sensitivity Oracles} 

\authorrunning{K. Harada, N. Kitamura, T. Izumi, T. Masuzawa}

\author{Kaito Harada}{Osaka University, Suita, Japan}{k-harada@ist.osaka-u.ac.jp}{}{}

\author{Naoki Kitamura}{Osaka University, Suita, Japan}{n-kitamura@ist.osaka-u.ac.jp}{}{}

\author{Taisuke Izumi}{Osaka University, Suita, Japan}{t-izumi@ist.osaka-u.ac.jp}{}{}

\author{Toshimitsu Masuzawa}{Osaka University, Suita, Japan}{masuzawa@ist.osaka-u.ac.jp}{}{}

\Copyright{Kaito Harada, Naoki Kitamura, Taisuke Izumi, and Toshimitsu Masuzawa} 



\funding{This work was supported by MEXT/JSPS KAKENHI Grant Number 22K21277, 23K16838, 23K24825, and 23H04385.}

\EventEditors{Timothy Chan, Johannes Fischer, John Iacono, and Grzegorz Herman}
\EventNoEds{4}
\EventLongTitle{32nd Annual European Symposium on Algorithms (ESA 2024)}
\EventShortTitle{ESA 2024}
\EventAcronym{ESA}
\EventYear{2024}
\EventDate{September 2--4, 2024}
\EventLocation{Royal Holloway, London, United Kingdom}
\EventLogo{}
\SeriesVolume{308}
\ArticleNo{91}

\begin{document}

\maketitle

\begin{abstract}
An \emph{$\alpha$-approximate vertex fault-tolerant distance sensitivity oracle} (\emph{$\alpha$-VSDO}) for 
a weighted input graph $G=(V, E, w)$ and a source vertex $s \in V$ is the data structure answering 
an $\alpha$-approximate distance from $s$ to $t$ in $G-x$ for any given query $(x, t) \in V \times V$. It is a data structure version of the so-called single-source replacement path problem (SSRP). In this paper, we present a new \emph{nearly linear-time} algorithm of constructing
a $(1 + \eps)$-VSDO for any directed input graph with polynomially bounded integer edge weights. More precisely, 
the presented oracle attains 
$\tilde{O}(m \log (nW)/ \eps + n \log^2 (nW)/\eps^2)$\footnote{The $\tilde{O}(\cdot)$ notation omits polylogarithmic 
factors, i.e., $\tilde{O}(f(n)) = O(f(n) \mathrm{polylog}(n))$.} construction time, $\tilde{O}(n \log (nW) / \eps)$ size\footnote{Throughout this 
paper, size is measured by the number of words. One word is $O(\log n)$ bits.}, 
and $\tilde{O}(1/\eps)$ query time, where $n$ is the number of vertices, $m$ is the number of edges, and $W$ is the maximum 
edge weight. These bounds are all 
optimal up to polylogarithmic factors. To the best of our knowledge, this is the first non-trivial 
algorithm for SSRP/VSDO beating $\tilde{O}(mn)$ computation time for directed graphs with general edge weight functions, 
and also the first nearly linear-time construction breaking approximation factor 3. Such a construction 
has been unknown even for undirected and unweighted graphs.
In addition, our result implies that the known conditional lower bounds for the exact SSRP computation does 
not apply to the case of approximation.

\end{abstract}

\newpage

\section{Introduction}
\subsection{Background and Our Result}
\label{sec:background}

The fault of links and vertices is ubiquitous in real-world networks. In fault-prone networks, it is important to develop an algorithm that quickly recompute the desired solution for a given fault pattern.  For example, suppose that a client wants to send information to all guests in the network, but the shortest paths can be disconnected by faults. Then, it needs to find the shortest paths avoiding all failed entities, so-called \emph{replacement paths}.
In this paper, we consider the \emph{Single Source Replacement Path Problem} (\emph{SSRP}) for directed graph $G=(V(G), E(G))$
with positive edge weights, which requires us to find the single-source shortest paths for all possible single vertex-fault patterns. That is,
its output is the distances from a given source vertex $s$ to any vertex $t \in V(G)$ in graphs $G - x$ for every $x \in V(G)$. This problem is known as one of the fundamental problems, not only in the context of fault-tolerance, but 
also in the auction theory~\cite{HS01,NR99}. While SSRP is also considered in the edge-fault case, it is easily reduced to the vertex-fault case of SSRP. 
Hence this paper only focuses on 
vertex faults.

A trivial algorithm for SSRP is to solve the single-source shortest path problem on $G-x$ for each failed vertex 
$x \in V(G) \setminus \{s\}$.
The running time of this algorithm is $\tilde{O}(mn)$, where $n$ is the number of vertices and $m$ 
is the number of edges. While it is an intriguing question whether non-trivial speedup of computing SSRP is possible or not, a variety of conditional lower bounds has been presented (see Table~\ref{table:lowerbound}). Focusing on directed graphs
with arbitrary positive edge weights, $\Omega(mn)$ time lower bound in the path comparison model~\cite{KKP93} has been proved by Hershberger, Suri, and Bhosle~\cite{HSB07}. A similar
fine-grained complexity barrier is also provided by Chechik and Cohen~\cite{CC19}, which exhibits $\tilde{\Omega}(mn^{1 - \delta})$ lower bound for any small constant 
$\delta > 0$ under the assumption that there exists no boolean matrix multiplication algorithm 
running in $O(mn^{1 - \delta})$ time for $n \times n$ matrices with a total number of $m$ $1$s. Due to these
results, the quest for much faster algorithms essentially requires some relaxation on the
problem setting. The research line considering undirected and/or unweighted graphs recently
yields much progress. For directed and unweighted graphs, a randomized SSRP algorithm running in $\tilde{O}(m\sqrt{n} + n^2)$ time has been shown~\cite{CM20}, which is also  
proved conditionally tight under some complexity-theoretic hypotheses. 

In this paper, we aim to circumvent the barriers above by admitting approximation. 
While such an approximation approach is known and succeeded 
for the $s$-$t$ replacement path problem (RP)~\cite{Ber10}, there have been no results so far for SSRP. A trivial 
bottleneck of SSRP is $O(n^2)$ term to output the solution containing the distances for all possible pairs of 
a failed vertex and a target vertex. This bottleneck slightly spoils the challenge to approximate SSRP algorithms 
of $o(m\sqrt{n}) + O(n^2)$ running time because it benefits only when $m = \omega(n^{3/2})$ holds. Hence this paper 
focuses on the fast construction of the \emph{$\alpha$-approximate vertex fault-tolerant distance sensitivity oracle} 
(\emph{$\alpha$-VSDO}), which is the compact (i.e. $o(n^2)$ size) data structure answering an $\alpha$-approximate 
distance from $s$ to $t$ in $G-x$ for any query $(x, t)$. This is naturally regarded as the data structure version 
of SSRP, and due to its compactness, the time for outputting the solution does not matter anymore. 
There are several results on the construction of $\alpha$-DSOs. 
However, in the case of $\alpha = (1 + \eps)$, 
no construction algorithm running faster than $O(m\sqrt{n})$ time is invented, even for undirected and unweighted graphs (see Tables~\ref{table:dso} and~\ref{table:rp}). The results explicitly faster
than $O(m\sqrt{n})$ time are by Baswana and Khanna~\cite{BK13} and Bil{\`o}, Guala, Leucci, and Proietti~\cite{BGLP18}, which attain only $3$-approximation for undirected graphs. Our main contribution is to solve this open problem positively and almost completely. 
The main theorem is stated as follows:
\begin{theorem} \label{thm:main}
    Given any directed graph $G$ with edge weights in range $[1, W]$, a source vertex $s \in V(G)$, and a constant $\eps \in (0, 1]$, there exists a deterministic algorithm of constructing a $(1+\eps)$-VSDO of size $O(\eps^{-1} n \log^3 n \cdot \log (nW))$. The construction time and the query processing time are respectively 
$O(\eps^{-1} \log^4 n \cdot \log (nW) (m + n \eps^{-1} \cdot \log^3 n \cdot \log (nW)))$ and 
    $O(\log^2 n \cdot \log (\eps^{-1} \log (nW)))$.
\end{theorem}
For polynomially-bounded edge weights and $\epsilon = \Omega(1/\mathrm{polylog}(n))$, one oracle attains $\tilde{O}(m)$ 
construction time, $\tilde{O}(n)$ size, and $\tilde{O}(1)$ query processing time, which are optimal up to polylogarithmic factors.
It also deduces an algorithm for $(1+ \eps)$-approximate SSRP running in $\tilde{O}(m + n^2)$ trivially. To the best of our knowledge, this is the first non-trivial result for SSRP/VSDO 
beating $O(mn)$ computation time for directed graphs with polynomially-bounded positive edge weights, and also the first 
nearly linear-time construction of breaking approximation factor 3. Such a construction has been unknown even for undirected 
and unweighted graphs.

We emphasize that nearly linear time construction of approximate DSOs inherently faces the  
challenge that oracles must be built \emph{without explicitly computing SSRP}. If $\mathrm{poly}(n)$ 
construction time is admitted, one can adopt the two-phase approach which 
computes the SSRP at first, and then compresses the computation result into a small-size oracle. 
In fact, this approach is adopted by most of known $(1 + \epsilon)$-approximate constructions.
It is, however, impossible 
in considering $\tilde{O}(m)$-time construction.
The high-level structure of our construction algorithm is the combination of the divide-and-conquer approach used in 
Grandoni and Vassilevska Williams \cite{GV12} and Chechik and Magen\cite{CM20}, and the \emph{progressive Dijkstra} 
algorithm by Bernstein~\cite{Ber10} originally designed for solving the approximate 
$s$-$t$ replacement path problem. The crux of our result is to demonstrate that the progressive Dijkstra provides 
much useful information for computing approximate SSRP beyond RP, by carefully installing it into the approach of \cite{CM20,GV12}. 
In addition, as a by-product, we refine the original progressive Dijkstra algorithm 
into a simpler and easy-to-follow form, which is of independent interest and potentially useful for its future applications.

\begin{table}[t]
  \caption{Known Results for EDSO/VDSO. The column ``problem'' describes 
the target problem of each result, 
where the first additional character V/E represents the fault model (vertex fault/edge fault). The column ``input'' describes the type of graphs each result covers. In the notation
$X/Y$, $X$ is either D (directed) or U (undirected), and the $Y$ is either 
U (unweighted), $+$ (positive edge weight), or $\pm$ (arbitrary edge weight). 
The column ``apx.'' represents approximation factors. The symbol $L$ is 
the shorthand of $\eps^{-1} \log (nW)$. 
The dagger mark $\dagger$ implies some additional features not described in the table, which is explained in Section~\ref{subsec:relatedwork}.}
  \label{table:dso}
  \centering
  \small
  \begin{tabular}{ccccccc}
    \hline
    ref & problem & input & apx. & size & construction & query\\
    \hline
    DT02\cite{DT02} & EDSO$^{\dagger}$ & D/$+$ & 1 & $O(n^2\log n)$ & $\tilde{O}(mn^2)$ & $\tilde{O}(1)$ \\
    BK09\cite{BK09} & EDSO$^{\dagger}$ & D/$+$ & 1 & $O(n^2)$ & $\tilde{O}(mn)$ & $O(1)$ \\
    BK13\cite{BK13} & VDSO & U/$+$ & 3 & $O(n \log n)$ & $\tilde{O}(m)$ & $O(1)$ \\
    BK13\cite{BK13} & VDSO & U/U & $1+\eps$ & $O(\frac{n}{\eps^3} + n\log n)$ & $O(m\sqrt{\frac{n}{\eps}})$ & $O(1)$ \\
    BGLP16\cite{BGLP16} & EDSO & U/$+$ & 2 & $O(n)$ & $\tilde{O}(mn)$ & $O(1)$\\
    BGLP16\cite{BGLP16} & EDSO & U/$+$ & $1+\eps$ & $O(\frac{n}{\eps} \log \frac{1}{\eps})$ & $\tilde{O}(mn)$ & $O(\log n \cdot\frac{1}{\eps} \log \frac{1}{\eps})$\\
    GS18\cite{GS18} & ESDO$^{\dagger}$ & U/U & 1 & $\tilde{O}(n^{3/2})$ & $\tilde{O}(mn)$ &
    $\tilde{O}(1)$ \\
    BGLP18\cite{BGLP18} & EDSO & U/$+$ & 3 & $O(n)$ & $\tilde{O}(m)$ & $O(1)$ \\
    BGLP18\cite{BGLP18} & EDSO & U/U & $1+\eps$ & $O(\frac{n}{\eps^3})$ & $O(m\sqrt{\frac{n}{\eps}})$ & $O(1)$ \\
    BCHR20\cite{BCHR20} & VDSO & D/$+$ & $1+\eps$ & $O(nL)$ & $O(mn)^{\dagger}$ & $O(\log L)$\\
    BCFS21\cite{BCFS21} & ESDO & U/U & 1 & $\tilde{O}(n^{3/2})$ & $\tilde{O}(m\sqrt{n} + n^2)$ & $\tilde{O}(1)$ \\
    
    DG22\cite{DG22} & ESDO & U/U & 1 & $\tilde{O}(n^{3/2})$ & $\tilde{O}(m\sqrt{n})$ &
    $\tilde{O}(1)$ \\

    \textbf{This paper} & VDSO & D/$+$ & $1+\eps$ & $O(nL\log^3 n)$ & $\tilde{O}(mL + nL^2)$ & $O(\log^2 n \cdot \log L)$\\
    \hline
  \end{tabular}

\vspace{3mm}

  \caption{Known Upper Bounds for RP and SSRP, where $\omega$ is the matrix multiplication exponent, and $\alpha(n, m)$ is the inverse of the Ackermann function.
}
  \label{table:rp}
  \centering
  \small
  \begin{tabular}{ccccc}
    \hline
    ref & problem & input & apx. & time \\
    \hline
    MMG89~\cite{MMG89} & ERP & U/$+$ & 1 & $O(m + n\log n)$ \\
    NPW03~\cite{NPW03-1} & VRP & U/$+$ & 1 & $O(m + n\log n)$ \\
    NPW03~\cite{NPW03-2} & ERP & U/$+$ & 1 & $O(m\alpha(n,m))$ \\
    Ber10~\cite{Ber10} & ERP & D/$+$ & $1 + \eps$ & $\tilde{O}(m \log L)$ \\
    Vas11~\cite{Vas11} & ERP & D/$\pm$ & 1 & $\tilde{O}(Wn^{\omega})$ \\
    GV12~\cite{GV12} & ESSRP  & D/$\pm$ & $1$ & $\tilde{O}(W^{\frac{1}{4 - \omega}}n^{2 + \frac{1}{4 - \omega}})$ \\
    GV12~\cite{GV12} & ESSRP  & D/$+$ & $1$ & $\tilde{O}(Wn^{\omega})$ \\
    RZ12~\cite{RZ12} & ERP   & D/U & 1 & $O(m\sqrt{n})$ \\
    CC19~\cite{CC19}  & ESSRP & U/U & 1 & $\tilde{O}(m\sqrt{n} + n^2)$ \\
    CM20~\cite{CM20}  & ESSRP & D/U$^{\dagger}$ & 1 & $\tilde{O}(m\sqrt{n} + n^2)$ \\
    GPWX21~\cite{GPWX21} & ESSRP & D/$\pm$ & 1 & $\tilde{O}(M^{\frac{5}{17 - 4 \omega}}n^{\frac{36-7\omega}{17 - 4 \omega}})$ \\
    \hline
  \end{tabular}
\end{table}

\begin{table}[t]
  \caption{Known conditional lower bounds for RP and SSRP. All results apply to the case of exact computation.}
  \label{table:lowerbound}
  \centering
  \small
  \begin{tabular}{ccccc}
    \hline
    ref & problem & input & time & model/hypothesis \\
    \hline
    HSB07~\cite{HSB07} & ERP  & D/$+$ & $\Omega(m\sqrt{n})$ & Path comparison model \\
    HSB07~\cite{HSB07} & ESSRP   & D/$+$ & $\Omega(mn)$ & Path comparison model \\
    VW18~\cite{VW18} & ERP & D/$\pm$ & $\Omega(n^{3- \delta})\log W$ & no subcubic APSP \\
    VW18~\cite{VW18} & ERP & D/U & $\Omega(mn^{1 /2 - \delta})$ & no subcubic BMM \\
    CC19~\cite{CC19}  & ESSRP & U/U & $\Omega(mn^{1/2 - \delta} + n^2)$ & no $O(mn^{1- \delta'})$ BMM \\
    CC19~\cite{CC19}  & Comb. E-SSRP & U/$+$ & $\Omega(mn^{1-\delta})$ & no Comb. subcubic APSP \\
    CC20~\cite{CM20}  & ESSRP & U/$+$ & $\Omega(mn^{1/2 - \delta} + n^2)$ & no subcubic APSP \\
    \hline
  \end{tabular}
\end{table}

\subsection{Related Work}
\label{subsec:relatedwork}

The known results directly related to our algorithm are summarized in Tables~\ref{table:dso}, 
\ref{table:rp}, and \ref{table:lowerbound}. We explain some supplementary remark on 
these tables. The BCHR20 construction~\cite{BCHR20} is only the $(1 + \eps)$-approximation result which does not 
rely on the explicit SSRP computation. However, the construction is based on the information hard to compute 
in $\tilde{O}(m)$ time (the information $\mathit{HD}_{\alpha}(v)$ presented at Section 3 in~\cite{BCHR20}). While 
the BCHR20 does not state the construction time explicitly, it requires $O(mn)$ time following 
the authors' observation. The SSRP algorithm by Chechik and Magen~\cite{CM20} 
is referred to as the one for unweighted graphs, but it can also cover bounded rational 
edge weights
in the range $[1, c]$ with a constant $c > 1$. It is not easy to transform 
this algorithm into an approximate SSRP handling general edge weights by scaling 
approaches, because edge weights are \emph{lower bounded} by one. Some of the results in 
Table~\ref{table:dso}~\cite{BK09,DT02,GS18} actually provide the more stronger oracles 
which support all-pair or multiple-source queries.

There are a variety of topics closely related to SSRP/DSO. 
A \emph{$f$-fault-tolerant (approximate) shortest path tree} is a subgraph of the input graph 
$G$ which contains a (approximate) shortest path tree in $G - F$ for any faulty edge set $F$ of size at most $f$~\cite{BGLP22,PP18,PP20}. Any $1$-fault tolerant shortest path tree can be 
seen as a graph representation of the output of (approximate) SSRP, and its construction
is closely related to the construction of DSOs. In fact, some of the results in 
Table~\ref{table:dso} also include the construction of fault-tolerant shortest path 
trees~\cite{BCHR20,BGLP18}. 

The oracles and algorithms covering multiple failures are also
investigated~\cite{BCCCFKS23,CCFK17,CZ24,DGR21,DR22,WWX22}. The most general structure 
is the \emph{$f$-sensitivity all-pairs oracle}, which returns an exact or approximate 
length of the shortest replacement path avoiding a given faulty edge set $F$ such that 
$|F| \leq f$ holds. Currently, most of oracles covering general $f$ faults assume
undirected graphs, and the construction time and the oracle size are essentially 
more expensive than single-source 1-sensitivity oracles.

\section{Preliminary}
\label{sec:preliminary}


Let $V(H)$ and $E(H)$ denote the vertex and edge set of a directed graph $H$ respectively, and
$w_H(u, v)$ or $w_H(e)$ denote the weight of edge $e = (u, v)$ in $H$.
Let $G$ be the input directed graph of $n$ vertices and $m$ edges with integer edge weights within $[1, W]$. We assume
edge weights are polynomially bounded, i.e., $W = \mathrm{poly}(n)$. Since we only consider approximate shortest paths, 
the assumption of integer edge weights is not essential, because real weights can be rounded with an appropriate precision.

For a vertex set $U \subseteq V(H)$ of a graph $H$, $H - U$ denotes the graph obtained by removing the vertices in $U$ and the edges incident to them from $H$. Similarly, for an edge set $U \subseteq E(H)$, $H- U$ denotes the graph obtained by removing the edges
in $U$. We often use the abbreviated notation $H - u$ when $U$ consists of a single element $u$. Let $N_H(u)$ be the set of 
the neighbors of $u \in V(H)$ in $H$.
For a vertex subset $U \subseteq V(H)$, let $\Iset(U)$ be the set of the edges incident to a vertex $u \in U$. Note that $\Iset(U)$ contains the edges whose endpoints are both in $U$. Given a set $X$ 
of vertices or edges of $H$, we define $H[X]$ as the subgraph induced by $X$.

A sequence of vertices $A = \Path{a_0, a_1, \dots, a_{k-1}}$ such that $(a_i, a_{i+1}) \in E(H)$ $(0 \leq i < k-1)$ is called a path from $a_0$ to $a_{k-1}$ in $H$. We use the terminology ``path'' for a simple path (i.e., $a_i$ for $0 \leq i \leq k-1$ are all different). For simplicity, we often deal with a path $A$ as the path subgraph corresponding to $A$.
We define $A[i,j]$ as its sub-path $\Path{a_i, a_{i+1}, \dots, a_j}$, and $w(A)$ as the weighted length of $A$.
The (directed) distance 
$\Dist_H(u, v)$ from $u$ to $v$ in $H$ is the length of the $u$-$v$ shortest path in $H$.
For any set $\Sigma$ of paths, we define $\Minlength(\Sigma)$ as $\Minlength(\Sigma) = \min\{w(Q) \mid Q \in \Sigma\}$. 
We also define $\Minpath(\Sigma)$ as the path $Q \in \Sigma$ of length $\Minlength(\Sigma)$. If two or more paths have
length $\Minlength(\Sigma)$, an arbitrary one of them is chosen as a canonical path.

The problem considered in this paper is defined as follows.
\begin{definition}[$(1+\eps)$-VSDO]
    A $(1+\eps)$-VSDO for a directed graph $G$ with positive weights, a single source $s \in V(G)$, and a positive constant $\eps \in (0, 1]$ is the data structure supporting query $\Query(v_f, x)$ for any failed vertex $v_f \in V(G)\setminus \{s\}$ and destination $t \in V(G)$ which returns a value satisfying $\Dist_{G-v_f}(s, t) \leq \Query(v_f, x) \leq (1+\eps) \cdot \Dist_{G-v_f}(s, t)$.
\end{definition}

\section{Technical Outline} \label{sec:outline}

In our construction, we recursively construct the oracle based on the \emph{centroid bipartition}, which is the approach also 
adopted in~\cite{GV12} and \cite{CM20}. Given a tree $T'$ of $n$ vertices, a \emph{centroid} $z \in V(T')$ of $T'$ is the vertex 
in $V(T')$ such that any connected component of $T -z$ contains at most 
$n/2$ vertices. Using the centroid $z$, one can split $T'$ into two edge-disjoint connected subtrees $T'_1$ and $T'_2$ of $T'$ such that $E(T'_1) \cup E(T'_2) = E(T')$, $V(T'_1) \cap V(T'_2) = \{z\}$, and $\frac{n}{3} \leq |V(T'_1)|, |V(T'_2)| \leq \frac{2n}{3}$, which we call the \emph{centroid bipartition} of $T'$. If $T'$ is a rooted tree,
we treat both $T'_1$ and $T'_2$ also as rooted trees. The split tree containing the original root $s$ is referred to as $T'_1$,
whose root is $s$, and the other one is referred to as $T'_2$, whose root is $z'$.
It is easy to find the centroid bipartition $T'_1$ and $T'_2$ from $T'$ in $O(n)$ time.
We also denote the $s$-$z$ path in $T'$ by $P_{T'}$.

Given the input graph $G$ and a source vertex $s$, our algorithm first constructs a shortest path tree $T$ rooted by $s$, and its centroid bipartition $T_1$ and $T_2$.
The whole oracle for $G$ and $s$ consists of three sub-oracles. Given a query 
$(x, f) \in (V(G) \setminus \{s\}) \times V(G)$, they respectively handle the 
three cases below (see Fig. \ref{divided case}).  
\begin{enumerate}
    \item $x \in V(P_T)$ and  $t \in V(T_2) \setminus \{z\}$
    \item $x \in V(T_2) \setminus \{z\}$ and $t \in V(T_2) \setminus \{z\}$
    \item $x \in V(T_1)$ and $t \in V(T_1)$
\end{enumerate}
Notice that, in the first case, we use $x \in V(P_T)$ instead of $x \in V(T_1)$; this is only the non-trivial situation because $\Dist_G(s, t) = \Dist_{G - x}(s, t)$ obviously holds if $x \not\in V(P_T)$. The sub-oracles for the second and third cases are recursively constructed for the graphs $G_2$ and $G_1$
obtained by slightly modifying $G[V(T_2)]$ and $G[V(T_1)]$. 
Since the modification does not increase the sizes of $G_1$ and $G_2$ so much from $G[V(T_1)]$ and $G[V(T_2)]$, 
one can guarantee that the recursion depth is logarithmic. Hence the total construction time
is easily bounded by $\tilde{O}(m)$ if the times of constructing the first-case sub-oracle, $G_1$, and $G_2$ are all $\tilde{O}(m)$. 

\begin{figure*}[tb]
    \centering
    \begin{minipage}[b]{0.3\linewidth}
        \centering
        \includegraphics[keepaspectratio, scale=0.2]{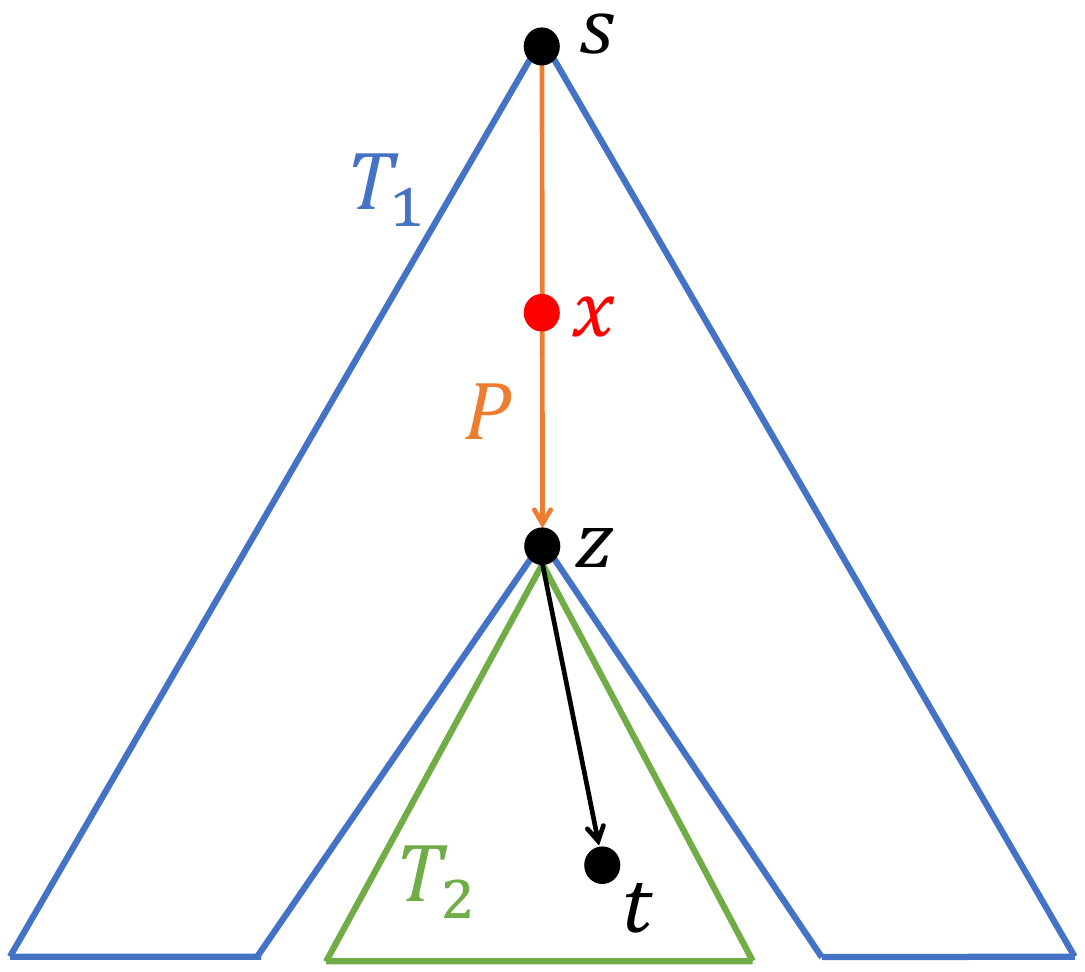}
        \subcaption{Case1}
    \end{minipage}
    \begin{minipage}[b]{0.3\linewidth}
        \centering
        \includegraphics[keepaspectratio, scale=0.2]{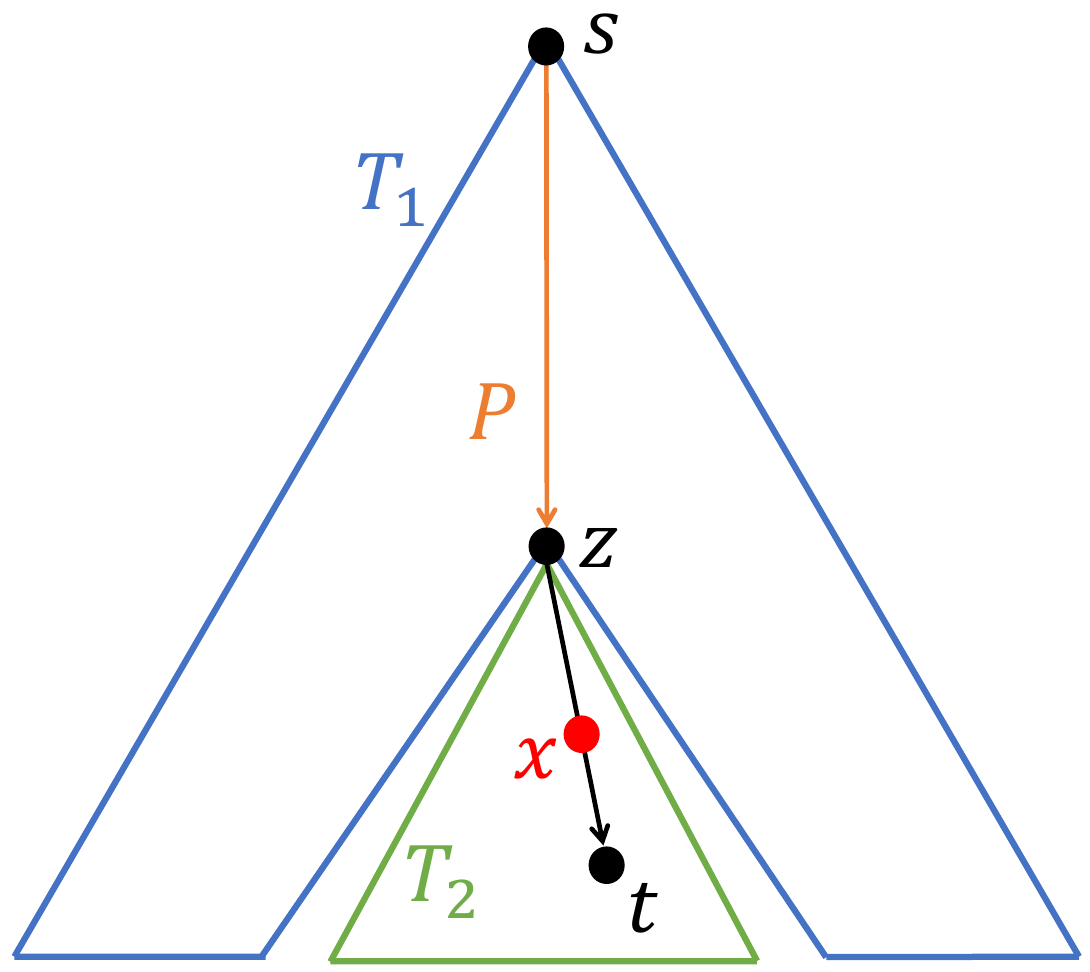}
        \subcaption{Case2}
    \end{minipage}
    \begin{minipage}[b]{0.3\linewidth}
        \centering
        \includegraphics[keepaspectratio, scale=0.2]{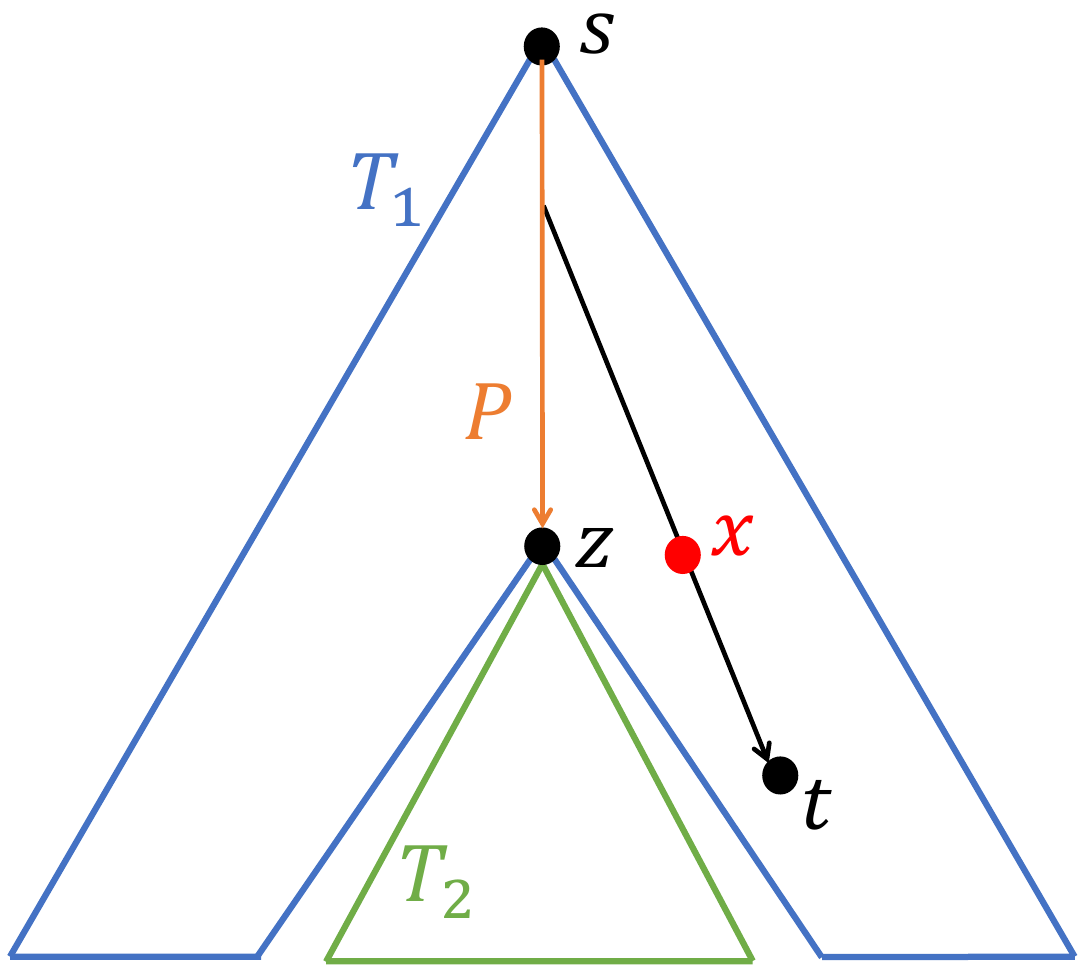}
        \subcaption{Case3}
    \end{minipage}
    
    \caption{Three cases in the oracle construction}
    \label{divided case}
\end{figure*}

\subsection{Sub-oracle for the First Case} \label{subsec:DP}

The precise goal for the first case is to develop a \emph{$Q$-faulty $(1 + \eps_1)$-VSDO}, which is a $(1 + \eps_1)$-approximate VSDO only supporting the faults on a given path $Q$, which is formally defined as follows:
\begin{definition}[$Q$-faulty $(1+\eps_1)$-VSDO] \label{p-faulty oracle def}
A $Q$-faulty $(1 + \eps_1)$-VSDO for a directed graph $G$ with positive weights, a source vertex $s \in V(G)$, and a path $Q$ from $s$ is the data structure supporting the query $\Fquery_Q(x, t)$ for any $x \in V(Q)$ and $t \in V(G)$ which returns a value satisfying $\Dist_{G-x}(s, t) \leq \Fquery_Q(x, t) \leq (1+\eps_1) \cdot \Dist_{G-x}(s, t)$.
\end{definition}
Obviously, the $P_T$-faulty $(1 + \eps_1)$-VSDO for $G$ is the sub-oracle covering the first case, where we need to set $\eps_1 = \eps/(3\log n)$ to deal with the accumulation of approximation error in the recursive construction (this point is explained later in more details). Let $P_T = \Path{v_0, v_1, \dots, v_{p-1}}$ $(s=v_0, z=v_{p-1})$, 
and $v_f$ be the alias of the faulty vertex $x \in V(P_T)$. The construction of 
the $P_T$-faulty $(1 + \eps_1)$-VSDO follows 
a generalized and simplified version of the approximate $s$-$t$ replacement path 
algorithm by Bernstein~\cite{Ber10}. We first introduce two types for $s$-$t$ paths 
in $G - v_f$ (see also Fig.~\ref{ex-dj}).
\begin{definition}(Jumping and departing paths) \label{dj def}
    For any failed vertex $v_f \in V(P_T) \setminus \{s\}$ and destination $t \in V(G)$, a $s$-$t$ path $R$ in $G - v_f$ is called a \emph{departing path} avoiding $v_f$ if it satisfies the following two conditions (C1) and (C2):
    \begin{itemize}
    \item (C1) The vertex subset $V(R) \cap V(P_T[0, f-1])$ induces a prefix of $R$.
    \item (C2) $(V(R) \setminus \{t\}) \cap V(P_T[f+1, p-1]) = \emptyset$.
    \end{itemize}
    A $s$-$t$ path $R$ in $G - v_f$ is called a \emph{jumping path} avoiding $v_f$ if it satisfies the condition (C1) and (C3) below:
    \begin{itemize}
        \item (C3) $(V(R) \setminus \{t\}) \cap V(P_T[f+1, p-1]) \neq \emptyset$.
    \end{itemize}
\end{definition}
We call the last vertex of the prefix induced by $V(R) \cap V(P_T[0, f-1])$ the \emph{branching vertex} of $R$, which we refer to as $b(R)$. In addition, 
for any jumping path $R$ avoiding $v_f$, we define $c(R)$ as the first vertex in $V(R) \cap V(P_T[f+1, p-1])$ 
(with respect to the order of $R$), which we refer to as the \emph{coalescing vertex} of $R$. 
Given any failed vertex $v_f \in V(P_T) \setminus \{s\}$ and destination $t \in V(G)$, $\DDist_{G - v_f}(s, t)$ and 
$\JDist_{G - v_f}(s, t)$ respectively denote the lengths of the shortest departing and jumping paths from $s$ to $t$ avoiding 
$v_f$.  Since any path satisfies either (C2) and (C3), and one can assume that any shortest $s$-$t$ path in $G - v_f$ satisfies 
(C1) without loss of generality, $\Dist_{G - v_f}(s, t) = \min\{\DDist_{G - v_f}(s, t), \JDist_{G - v_f}(s, t)\}$ holds. Hence it suffices to construct two data structures approximately answering the values of $\DDist_{G - v_f}(s, t)$ and $\JDist_{G - v_f}(s, t)$ for any $v_f \in V(P_T)$ and $t \in V(T_2)$ respectively. For $\DDist_{G - v_f}(s, t)$,
we realize it as an independent sub-sub-oracle:

\begin{definition}[DP-Oracle] \label{dp-oralce def}
A \emph{$Q$-faulty $(1+\eps_1)$-DPO (DP-oracle)} for a graph $G$, a source vertex $s \in V(G)$, and a path $Q$ from $s$ 
is the data structure supporting the query $\Dquery(x, t)$
for any $x \in V(Q)$ and $t \in V(G)$ which returns a value 
satisfying $\DDist_{G-x}(s, t) \leq \Dquery(x, t) \leq (1+\eps_1) \cdot \DDist_{G-x}(s, t)$.
\end{definition}
It should be noted that our definition of $P_T$-faulty $(1 + \eps_1)$-DPOs supports queries for all $t \in V(G)$, not limited to $V(T_2)$. This property is not necessary for addressing the first case, but used in other cases.

\begin{figure*}[tb]
    \centering
    \begin{minipage}[b]{0.3\linewidth}
        \centering
        \includegraphics[keepaspectratio, scale=0.2]{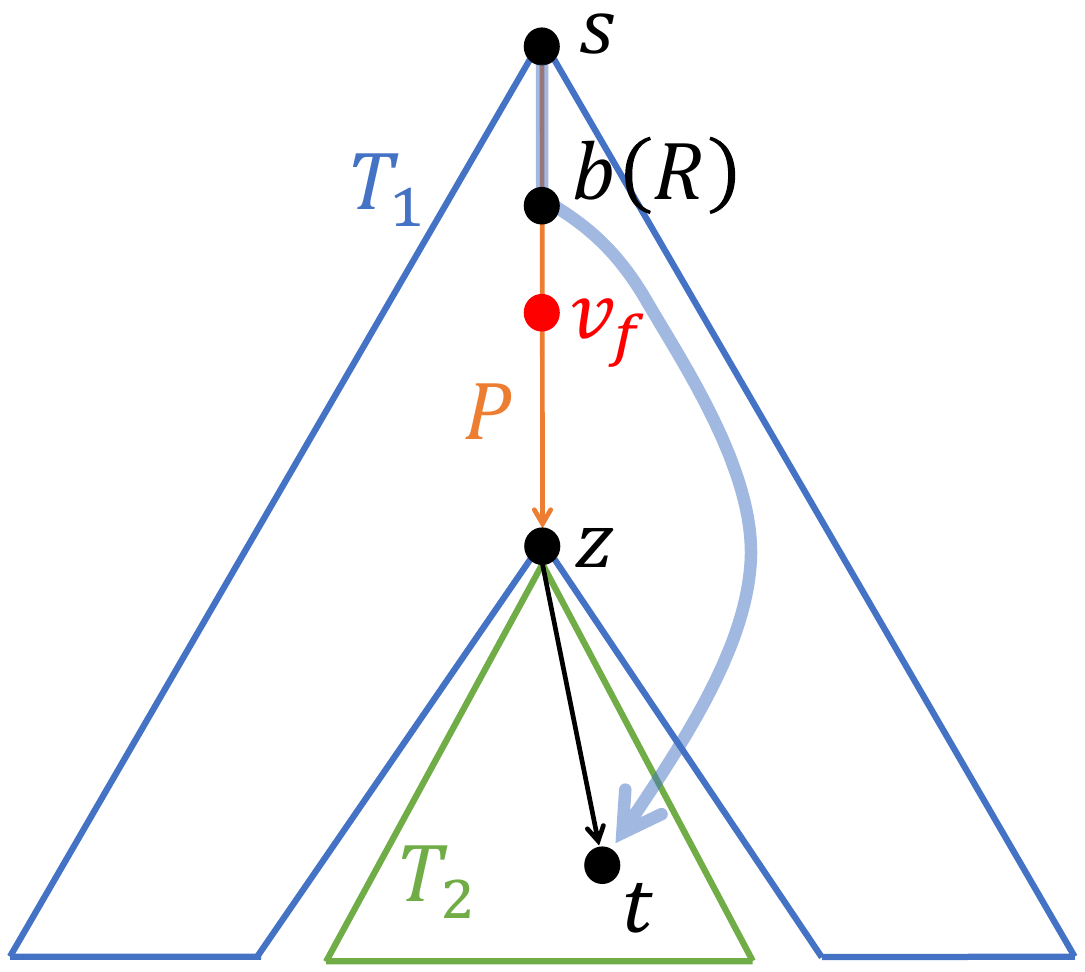}
        \subcaption{departing path}
    \end{minipage}
    \begin{minipage}[b]{0.3\linewidth}
        \centering
        \includegraphics[keepaspectratio, scale=0.2]{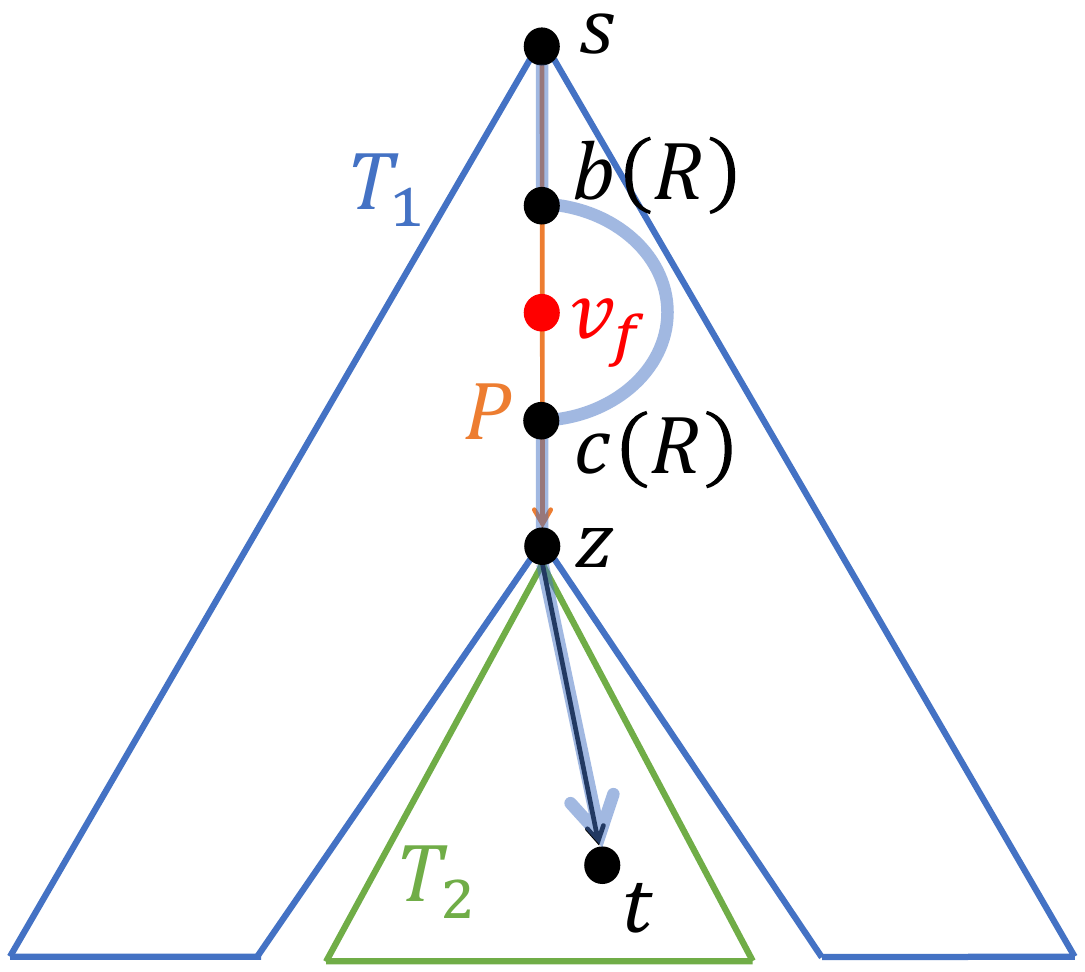}
        \subcaption{jumping path}
    \end{minipage}
    \caption{Examples of departing and jumping paths. $b(R)$ represents the branching vertex on each path, and $c(R)$ represents a coalescing vertex in (b).}
    \label{ex-dj}
\end{figure*}

The key technical ingredient of implementing DP-oracles is the \emph{progressive Dijkstra} algorithm presented in~\cite{Ber10}. We first introduce several notations and terminologies for explaining it: 
since $V(R) \cap V(P_T[0, f -1]) = V(R) \cap V(P_T)$ necessarily holds 
for any departing path avoiding $v_f$, the branching vertex $b(R)$ is the last vertex of the path induced by $V(R) \cap V(P_T)$, i.e., $b(R)$ is determined  
only by $R$, independently of $v_f$. It allows us to treat departing paths without association of avoiding vertex $v_f$. That is, the sentence ``a departing path $R$ with branching vertex $b(R)$'' means that $R$ is a departing path avoiding some successor of $b(R)$ (to which we do not pay attention). 
Assume that $p = |V(P_T)|$ is a power of $2$ for simplicity. Given $0 \leq i \leq \log p$ and $0 \leq j < 2^i$, we define $I^i_j$ as $I^i_j = 
P_T[p \cdot j/2^{i}, p \cdot (j + 1)/2^{i} - 1]$.
Intuitively, for each $i$, $I^i_0, I^i_1, \dots I^i_{2^i -1}$ form a partition of $P_T$ into $2^i$ sub-paths of length $p/2^i$ (see Fig.~\ref{partition}). For any sub-path $I$ of $P_T$, let $\Gamma(I, t)$ be the set of all $s$-$t$ departing paths $R$ such that 
$b(R) \in V(I)$ holds. Roughly, the progressive Dijkstra provides a weaker form of the approximate values of 
$\Minlength(\Gamma(I^i_j, t))$ for all $i$ and $j$ in $\tilde{O}(m / \eps_1)$ time, where ``weaker form'' means that it computes $\Minlength(\Gamma(I^i_j, t))$ such that $(1 + O(\eps_1 / \log n)) \cdot \Minlength(\Gamma(I^i_j, t)) < \min_{0 \leq j' < j} \Minlength(\Gamma(I^i_{j'}, t))$ holds. The intuition of this condition is that the progressive Dijkstra computes $\Minlength(\Gamma(I^i_j, t))$ only when it is not well-approximated by the length of the best departing path already found.
Letting $\mathcal{I}$ be the set of $I^i_j$ for all $i$ and $j$, for any failed vertex $v_f \in V(P_T)$, 
the prefix $P_T[0, f-1]$ is covered by a set $C \subseteq \mathcal{I}$ of at most $O(\log p)$ sub-paths in $\mathcal{I}$ (see the grey intervals of Fig.~\ref{partition}). Since the branching vertex of the shortest $s$-$t$ departing path avoiding $v_f$ obviously 
lies on $P_T[0, f-1]$, $\DDist_{G - v_f}(s, t) = \min_{I \in C} \Minlength(\Gamma(I, t))$ holds.
By a careful analysis, it is proved that this equality provides a $(1 + \eps_1)$-approximation 
of $\DDist_{G - v_f}(s, t)$ by using the outputs of the progressive Dijkstra as approximation 
of $\Minlength(\Gamma(I, t))$\footnote{The equality for computing $\DDist_{G - v_f}(s, t)$ actually used in the DP-oracle is further 
optimized using some additional properties of the progressive Dijkstra, which avoids 
the explicit construction of $C$ (see Lemma~\ref{lma:dporacle} for the details).}.

We explain how jumping paths are handled.
Let $R$ be the shortest $s$-$t$ jumping path avoiding $v_f \in V(P_T)$. Since the path from the coalescing vertex $c(b)$ to $t$ in $T$ (of length $\Dist_G(c(b), t)$) is not disconnected by $v_f$, one can assume that the suffix of $R$ from $c(b)$ is the path in 
$T$ without loss of generality. Then the suffix necessarily contains $z$, and thus it seems that 
$\JDist_{G - v_f}(s, t) = \Dist_{G - v_f}(s, z) + \Dist_T(z,t)$ holds.
Unfortunately, it does not always hold, because the shortest $s$-$z$ replacement path avoiding $v_f$ and the shortest $z$-$t$ path in 
$T$ might intersect. However, in such a case, $\DDist_{G - v_f}(s, t)  < \JDist_{G - v_f}(s, t)$ necessarily holds and thus the value 
from the DP-oracle well approximates $\Dist_{G - v_f}(s, t)$. Hence it is safe to use the above equality in our case.
For computing the approximate values of the right side of the equality for any $v_f \in V(P_T)$ and $t \in V(T_2) \setminus \{z\}$, it suffices to store the values $\Dist_{G - v_f}(s, z)$ for all $v_f \in V(P_T)$ approximately. We compute those values
using the Bernstein's algorithm~\cite{Ber10} which runs in $\tilde{O}(m/\eps_1)$ time.

\begin{figure*}[t]
    \centering
    \includegraphics[keepaspectratio, scale=0.3]{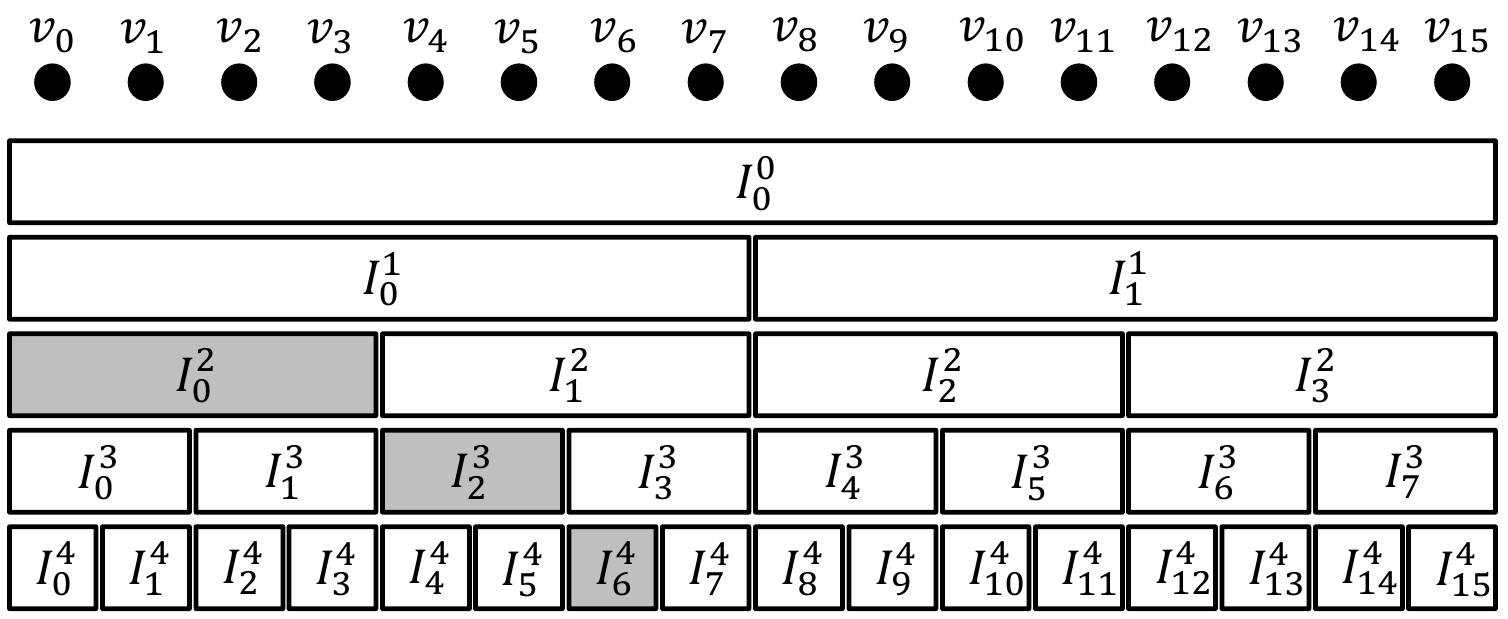}
    \caption{Example of partition of $P_T$ ($p=16$). The set of grey intervals is an example of $C$, for vertex $v_7$. }
    \label{partition}
\end{figure*}

\subsection{Construction of \texorpdfstring{$G_2$}{G2}}
The construction of $G_2$ for case 2 is relatively simple. Let $R$ be the shortest $s$-$t$ replacement path from $s$ to $t \in V(T_2)$ 
avoiding a failed vertex $x \in V(T_2)$. A key observation is that one can assume that 
$R$ does not contain any backward edge from $V(T_2) \setminus \{z\}$ to $V(T_1)$: if $R$ goes back from $V(T_2)$ to $V(T_1)$ 
through an edge 
$(u, v)$, one can replace the prefix of $R$ up to $v$ by the shortest path from $s$ to $v$ in $T_1$ (See Fig. \ref{natural-rp}), which never increases the length of $R$ because $s$-$v$ path in $T_1$ is the shortest path in $G$. 
Notice that the $s$-$v$ path in $T_1$ does not contain the failed vertex since $x\in V(T_2)\setminus \{z\}$.
This fact implies that $R$ consists of the prefix contained in $T_1$ and the suffix contained in $G_2$ which are concatenated by a forward edge from $V(T_1)$ to 
$V(T_2) \setminus \{z\}$. This observation naturally yields our construction of $G_2$, where all the vertices in 
$V(T_1) \setminus \{s\}$ are deleted, and $s$ and each vertex $u \in V(T_2) \setminus \{z\}$ is connected by edge $(s, u)$ of weight $w_{G_2}(s, u) = \Dist_{G[\Iset(V(T_1))]}(z, u)$ (See Fig. \ref{construction2}). 
The weights are computed just by running Dijkstra, which takes $\tilde{O}(m)$ time. Intuitively, the added edge $(s, u)$
corresponds to the prefix of $R$ up to the first vertex in $V(R) \cap V(T_2)$. 
It should be noted that the edges from $s$ are \emph{safely usable}, that is, the paths corresponding to the added edges do not contain the failed vertex $x$ as stated above. The precise construction of $G_2$ follow the procedure below:
\begin{enumerate}
    \item $G_2 \gets G[V(T_2) \setminus \{z\}]$
    \item Add $s$ to $G_2$.
    \item Calculate the distance $\Dist_{G[\Iset(V(T_1))]}(s, u)$ for each vertex $u \in V(T_2) \setminus \{z \}$ 
    by running the standard Dijkstra in $G[\Iset(V(T_1))]$.
    \item Add the edge $(s, u)$ of weights $\Dist_{G[\Iset(V(T_1))]}(s, u)$ to $G_2$ for each vertex $u \in V(T_2) \setminus \{z \}$ 
    satisfying $\Dist_{G[\Iset(V(T_1))]}(s, u) < \infty$.
\end{enumerate}
The sub-oracle for $G_2$ is a $(1+\eps_1)$-VSDO for $G_2$ and the newly added source $s$, which is constructed 
recursively. The correctness of the constructed sub-oracle relies on the following lemma.

\begin{lemma} \label{lma:G2-correctness}
    For any failed vertex $x \in V(T_2) \setminus \{z\}$ and destination $t \in V(T_2) \setminus \{z\}$, 
    $\Dist_{G_2-x}(s, t) = \Dist_{G-x}(s, t)$ holds.
\end{lemma}

\begin{figure*}[tb]
    \centering
    \begin{minipage}[b]{0.45\linewidth}
        \centering
        \includegraphics[keepaspectratio, scale=0.2]{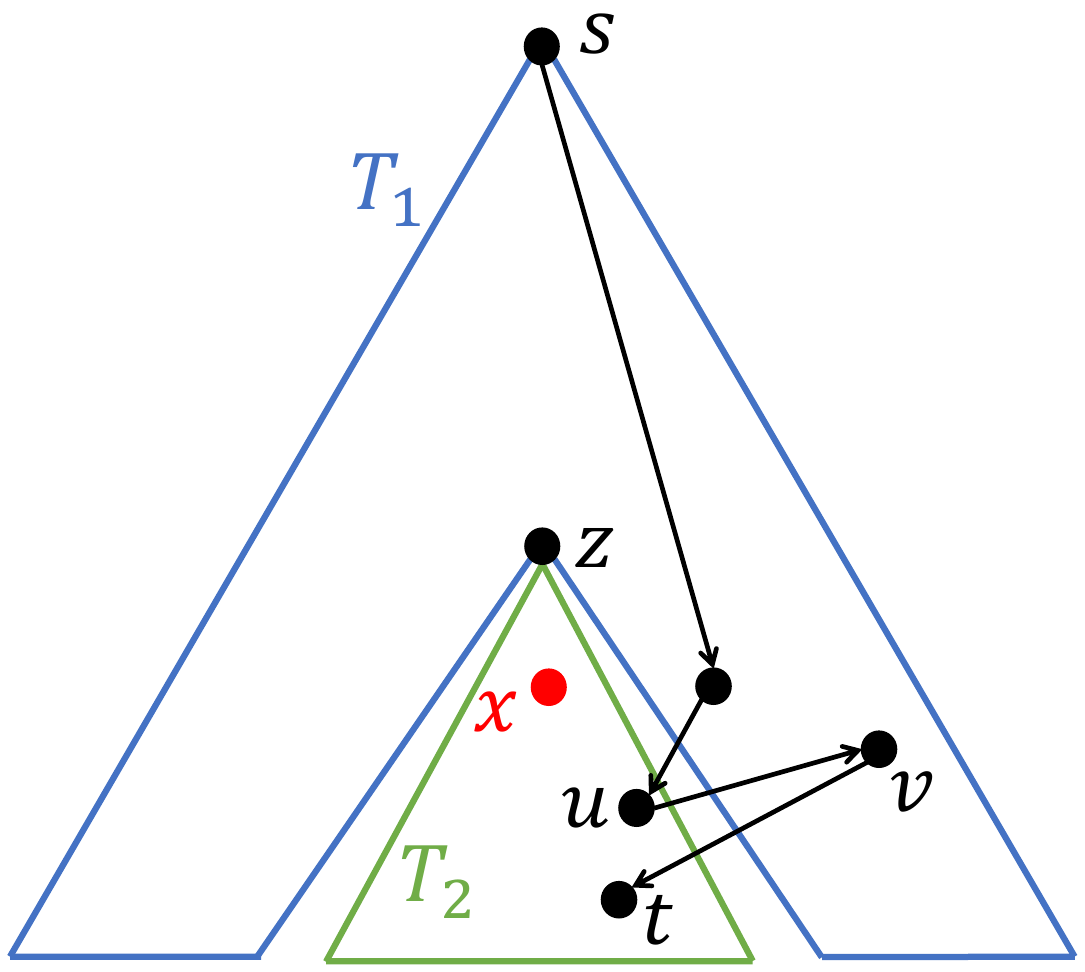}
        \subcaption{The original $R$ avoiding $x$}
    \end{minipage}
    \begin{minipage}[b]{0.45\linewidth}
        \centering
        \includegraphics[keepaspectratio, scale=0.2]{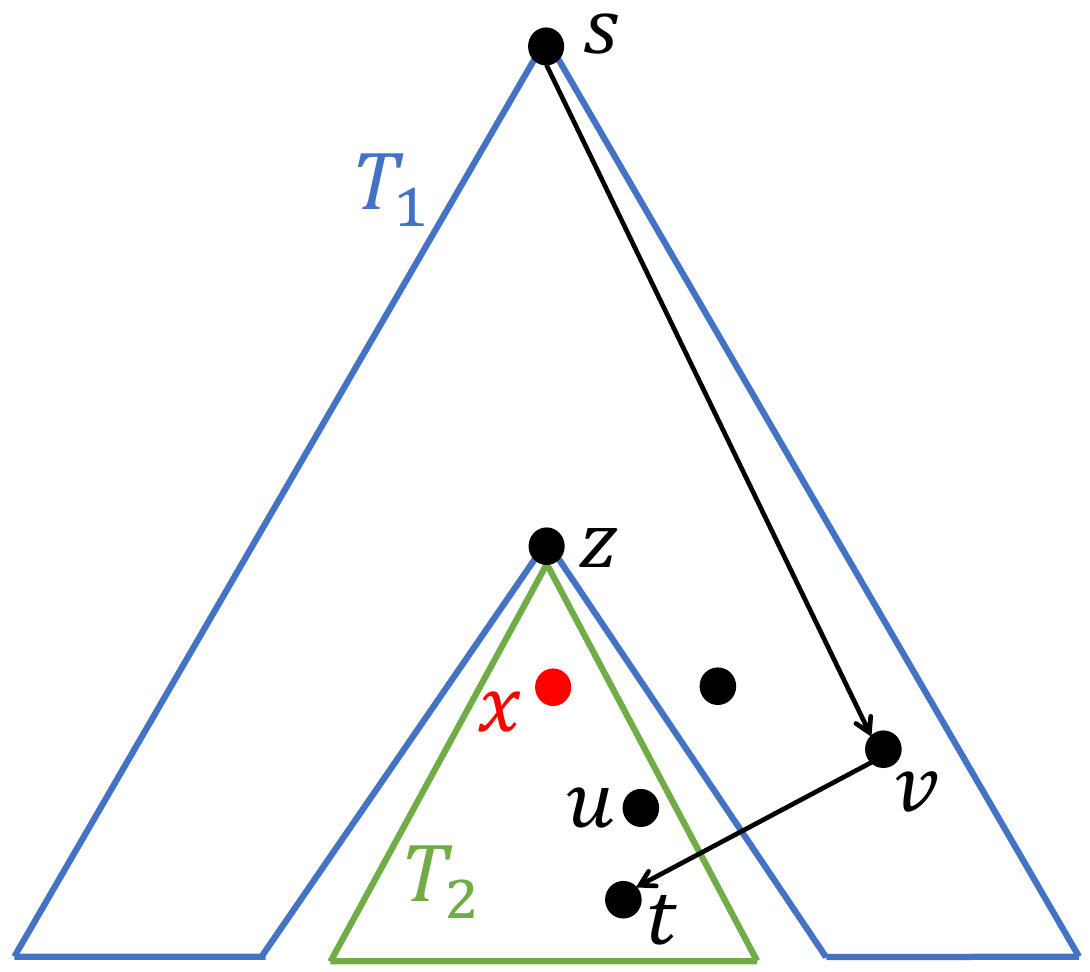}
        \subcaption{A shorter path avoiding $x$}
    \end{minipage}
    
    \caption{An Example of the Case that $E(R)$ contains a backward edge from $V(T_2)$ to $V(T_1)$.}
    \label{natural-rp}
\end{figure*}

\begin{figure*}[tb]
    \centering
    \begin{minipage}[b]{0.45\linewidth}
        \centering
        \includegraphics[keepaspectratio, scale=0.2]{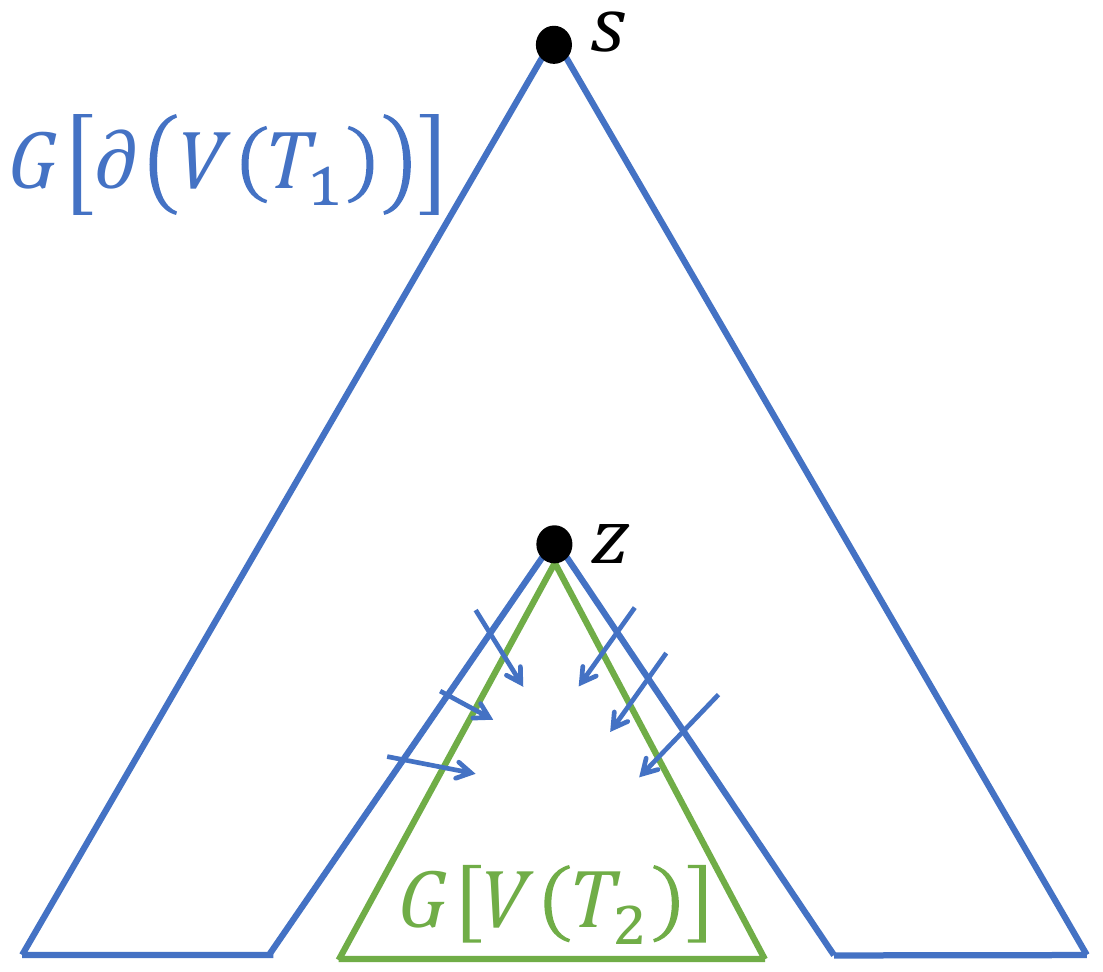}
        \subcaption{Graph $G[\Iset(V(T_2))]$ (blue subgraph)}
    \end{minipage}
    \begin{minipage}[b]{0.45\linewidth}
        \centering
        \includegraphics[keepaspectratio, scale=0.2]{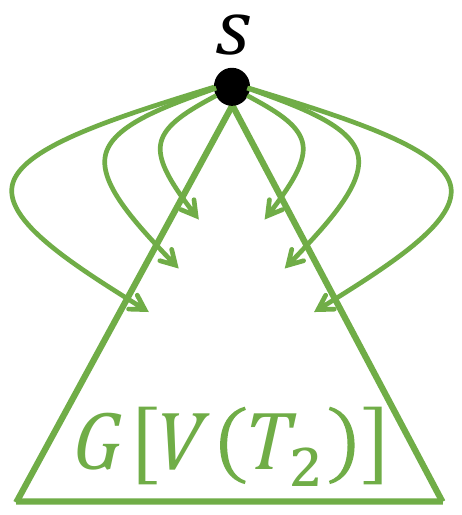}
        \subcaption{Graph $G_2$}
    \end{minipage}
    
    \caption{Graph $G[\Iset(V(T_2))]$ and the construction of $G_2$}
    \label{construction2}
\end{figure*}

\subsection{Construction of \texorpdfstring{$G_1$}{G1}}
\label{subsec:constG1}
We show the construction of $G_1$ for case 3. 
If the shortest $s$-$t$ replacement path avoiding $x$ does not cross 
$\overline{G}_1 = G[V(T_2) \setminus \{z\}]$, the recursive oracle for $G[V(T_1)]$ well approximates its length. 
However, paths crossing $\overline{G}_1$ are obviously omitted in such a construction. To handle
those omitted paths, we augment $G[V(T_1)]$ with an edge set $F$, i.e $G_1 = G[V(T_1)] + F$. Let $R$ be the shortest $s$-$t$ 
replacement path avoiding $x$, and assume that $R$ crosses $\overline{G}_1$. The more precise goal is to construct $F$ 
such that, for any sub-path(s) $R'$ of $R$ whose internal vertices are all in $V(\overline{G}_1)$, there exists a corresponding edge $e \in F$ 
from the first vertex of $R'$ to the last vertex of $R'$ with a weight approximately equal to $w(R')$. 
In our construction, the set $F$ is the union of the edge sets $F_1$ and $F_2$ provided 
by two augmentation schemes respectively addressing the two sub-cases of $x \not\in V(P_T)$ and $x \in V(P_T)$. 

Consider the first scheme which constructs $F_1$ coping with the sub-case of $x \not\in V(P_T)$. 
We denote by $(u, w)$ the last backward edge from $\overline{G}_1$ to $G[V(T_1)]$ in $R$ (with respect to the order of $R$).
One can assume that the sub-path of $R$ from $s$ to $u$ is the path of $T$ without loss of generality (recall that this 
sub-path does not contain $x$ by the assumption of $x \in V(T_1)$ and $x \not\in V(P_T)$). Then $R$ necessarily contains 
$z$. In addition, the length of the sub-path from $z$ to $w$ is equal to $\Dist_{G[\Iset(V(T_2))]}(z, w)$.
Hence this case is handled by adding to $G[V(T_1)]$ the edges $(z, u)$ for each 
$u \in V(T_1)$ of weight $w(z, u) = \Dist_{G[\Iset(V(T_2) \setminus \{z\})]}(z, u)$ (See Fig. \ref{construction1}).  
The weights of added edges are computed just by running Dijkstra with source $z$ in $G[\Iset(V(T_2))]$,
which takes $\tilde{O}(m)$ time. Since we call the recursive oracle for $G_1$ only when the 
failed vertex $x$ lies in $G_1$, the paths corresponding to the edges in $F_1$ are safely usable.

\begin{figure*}[tb]
    \begin{minipage}[b]{0.47\linewidth}
        \centering
        \includegraphics[keepaspectratio, scale=0.2]{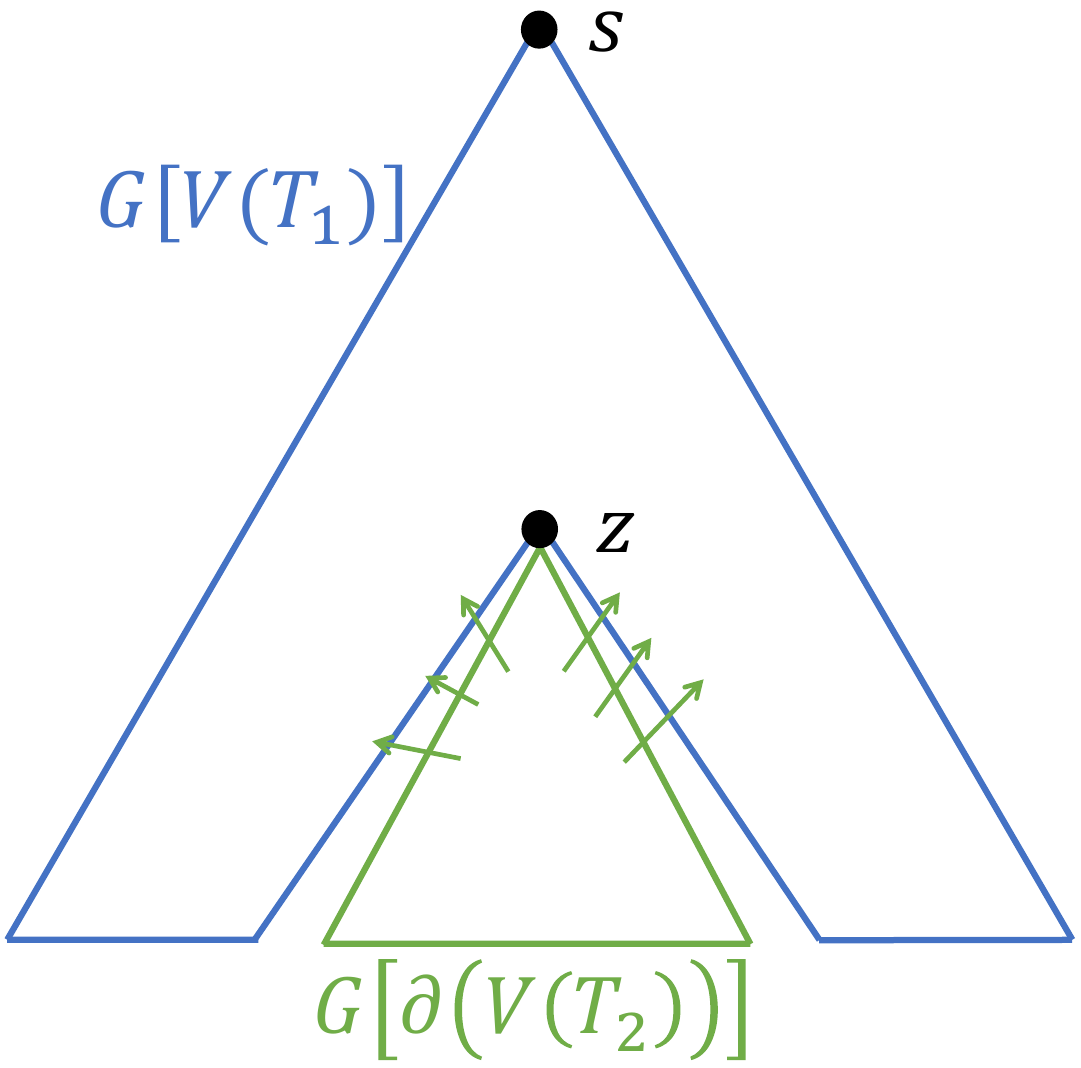}
        \subcaption{Graph $G[\Iset(V(T_2) \setminus \{z\})]$ (green subgraph)}
    \end{minipage}
    \begin{minipage}[b]{0.47\linewidth}
        \centering
        \includegraphics[keepaspectratio, scale=0.2]{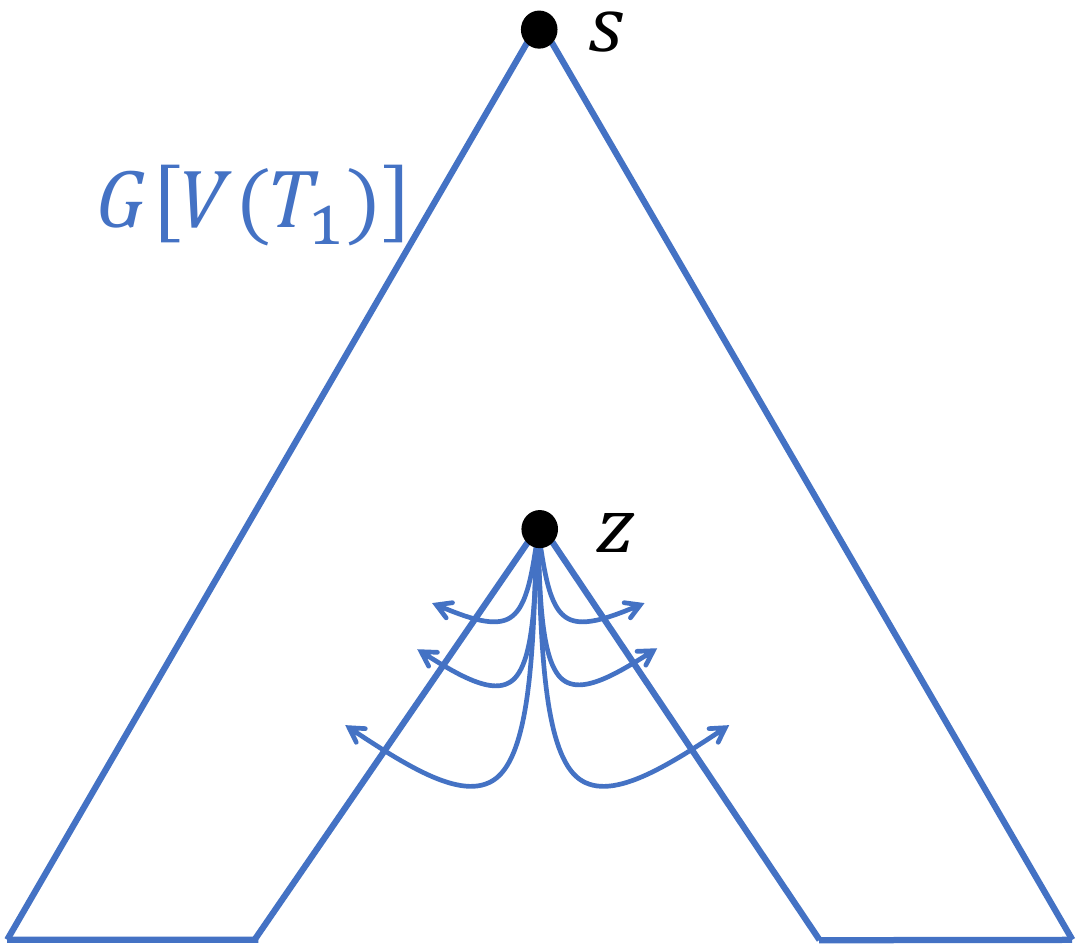}
        \subcaption{Graph $G[V(T_1)] + F_1$}
    \end{minipage}
    
    \caption{Construction of $F_1$}
    \label{construction1}
\end{figure*}

In the sub-case of $x(=v_f) \in V(P_T)$, we only have to consider the situation that $R$ is a jumping-path (i.e., $\JDist_{G - x}(s, t) < \DDist_{G - x}(s, t)$), because the DP-oracle 
constructed in the first case supports the query for any destination $t \in V(G)$, not only for $t \in V(T_2)$. 
Then there are the following two possibilities:
\begin{itemize}
\item The suffix from $c(R)$ contains a vertex in $V(\overline{G_1})$.
\item The prefix of $R$ up to $c(R)$ contains a vertex in $V(\overline{G_1})$.
\end{itemize}
Note that these two possibilities are not exclusive, and thus both can simultaneously happen.
The first possibility is handled by the augmentation with $F_1$. Let $(u, w)$ be the last backward edge from $\overline{G}_1$
to $G[V(T_1)]$ in $R$. Due to the case assumption of $\JDist_{G - x}(s, t) < \DDist_{G - x}(s, t)$, one can assume that the sub-path from 
$c(R)$ to $u$ is the path in $T$ containing $z$ (recall the argument of handling jumping paths in Section~\ref{subsec:DP}). 
The edge from $z$ to $w$ corresponding to the sub-path of $R$ from $z$ to $w$ has been already added as an edge in $F_1$. 
Hence it suffices to construct $F_2$ for coping with the second possibility. The baseline
idea is to add the edges corresponding to the sub-path of $R$ from $b(R)$ to $c(R)$. Unfortunately, a straightforward 
application of such an idea for all $R$ (over all possible combinations of $v_f$ and $t$) results in a too large cardinality of
$F_2$. Instead, we add only a few edges determined by the inside data structure of the DP-oracle which ``approximate'' the sub-paths from $b(R)$ to $c(R)$ for all possible $R$. The intuition behind this idea is the observation that the prefix of 
any jumping path $R$ up to $c(R)$ is a shortest departing path from $s$ to $c(R)$ with branching vertex $b(R)$. That is,
the $b(R)$-$c(R)$ sub-path of $R$ is the suffix of $\Minpath(\Gamma(I^i_j, c(R)))$ from its branching vertex 
for some $i$ and $j$ satisfying $b(R) \in I^i_j$. It naturally deduces the construction of $F_2$ by adding 
the edges corresponding to the suffixes of the paths computed by the progressive Dijkstra, because those paths 
well approximate $\Minpath(\Gamma(I^i_j, v))$ for all possible triples $(i, j, v)$.

The safe usability of the edges in $F_2$ is slightly delicate. Assuming the second sub-case of $x \in V(P_T)$,
any path $Q'$ corresponding to an edge in $F_2$ is certainly safe because $Q'$ is the suffix of a departing path $Q$ from 
$b(Q)$ and thus does not intersect $V(P_T)$. However, in the first sub-case of 
$x \not\in V(P_T)$, $Q'$ might be unavailable, 
although the corresponding edge in $F_2$ is still available. Then the recursive oracle for $G_1 = G[V(T_1)] + F$ can return an erroneous shorter length due to the existence of the edges in $F_2$ if the failed vertex $v_f$ does not lie in $P_T$. Fortunately, such an error never happens, because if $v_f$ does not lie on $P_T$, the whole of $P_T$ is available, and using the sub-path from $b(Q)$ to $c(Q)$ of $P_T$ always benefits more than using the added edge $(b(Q), c(Q))$. That is, one can assume that the shortest $s$-$t$ replacement path does not contain any edge in
$F_2$ without loss of generality if $v_f \not\in V(P_T)$ holds. 

In summary, the construction of $G_1$ follows the procedure below.

\begin{enumerate}
    \item $G_1 \gets G[V(T_1)]$.
    \item Calculate the distance $\Dist_{G[\Iset(V(T_2))]}(z, u)$ for each vertex $u \in V(T_1)$ by running the standard Dijkstra in $G[\Iset(V(T_2))]$.
    \item Add the edge $(z, u)$ of weights $\Dist_{G[\Iset(V(T_2))]}(z, u)$ to $G_1$ for each vertex $u \in V(T_1)$ satisfying $\Dist_{G[\Iset(V(T_2))]}(z, u) < \infty$.
    \item For each vertex $v_c \in V(P_T)$, $0 \leq i \leq \log p$ and $(\cdot, \ell, v_b) \in \mathsf{upd}(i, v_c)$, add the edge $(v_b, v_c)$ of weights $\ell - \Dist_T(s, v_b)$ to $G_1$.
\end{enumerate}

The sets of the edges added in step 3 and step 4 respectively correspond to $F_1$ and $F_2$.
If the constructed graph $G_1$ has parallel edges, only the one with the smallest weight is retained and the others are eliminated.
However, for simplicity of the argument, we assume that those parallel edges are left in $G_1$ without elimination.
The correctness of the sub-oracle recursively constructed for $G_1$ is guaranteed by the following lemma.

\begin{lemma} \label{lma:G1-correctness}
    For any failed vertex $x \in V(T_1)$ and destination $t \in V(T_1)$, the following conditions hold:
    \begin{itemize}
        \item If $x \not\in V(P_T)$, $\Dist_{G_1-x}(s, t) =  \Dist_{G-x}(s, t)$.
        \item If $x \in V(P_T)$ and $\JDist_{G -x}(s,t) < \DDist_{G - x}(s, t)$ holds, 
        $\Dist_{G - x}(s, t) \leq \Dist_{G_1-x}(s, t) \leq (1 + \eps_1) \cdot \Dist_{G -x}(s, t)$.
        \item Otherwise, $\Dist_{G-x}(s, t) \leq  \Dist_{G_1-x}(s, t)$.
    \end{itemize}
\end{lemma}
Notice that if neither of the first and second conditions are satisfied, the DP-oracle returns 
the correct approximate distance (recall the discussion in Section~\ref{subsec:constG1}). Then what the sub-oracle
for $G_1$ must avoid is to output an wrongly smaller value. It is ensured by the third condition.

Finally, we remark on the accumulation of approximation factors. Since the weights of edges in $F_2$ $(1 + \eps_1)$-approximate the length of the sub-paths from $b(R)$ to $c(R)$, the graph $G_1$ only guarantees that there exists a $s$-$t$ path avoiding $v_f$
whose length is $(1 + \eps_1)$-approximation of $w_G(R)$. Hence even if we construct an $(1 + \eps_1)$-VSDO for $G_1$,  
its approximation factor as a oracle for $G$ is $(1 + \eps_1)^2$. Due to $O(\log n)$ recursion depth, the approximation factor 
of our oracle finally obtained becomes $(1 + \eps_1)^{O(\log n)} = 1 + O(\eps_1 \log n)$, which is the reason why we need to set up 
$\eps_1 = \eps / (3 \log n)$.

\section{Whole Construction and Query Processing}
\label{sec:constAndQuery}

In this section, we summarize the whole construction of $(1 + \eps)$-VSDO, and present the details of the query processing.
As mentioned in Section~\ref{sec:outline}, we recursively construct the sub-oracles for $G_1$ and $G_2$ (for their new 
sources $s$).
The recursion terminates if the number of vertices in the graph becomes at most 6. At the bottom level, the query is processed
by running the standard Dijkstra, which takes $O(1)$ time. Let $G^1$ be the input graph, and $G^{2i}$ and $G^{2i+1}$ be the graphs 
$G_1$ and $G_2$ for $G^i$. We denote the shortest path tree of $G^i$ by $T_i$, its centroid by $z_i$, and the path from $s$ to $z_i$ in
$T_i$ by $P_i$. We introduce the \emph{recursion tree}, where each vertex is a graph $G^i$, and $G^{2i}$ and $G^{2i+1}$ are
the children of $G^i$ (if they exist). We define $\mathcal{G}_h = \{G^i \mid 2^h \leq i \leq 2^{h+1} - 1\}$, i.e., 
$\mathcal{G}_h$ is the set of $G^i$ with depth $h$ in the recursion tree. Since the centroid bipartition split $T_i$ into two edge disjoint 
subtrees of size at most $2n/3$, the recursion depth is bounded by $O(\log n)$, and the total number of recursive calls is $O(n)$. 
Since one recursive call duplicates the centroid $z$ (into $G_1$ and $G_2$), the number of vertices can increase. But the 
total increase is bounded by $O(n)$, and thus we have $\sum_{G' \in \mathcal{G}_h} |V(G')| = O(n)$. On the increase of edges, the 
number of edges added to $G_1$ is bounded by the number of forward edges from $V(T_1)$ to $V(T_2) \setminus \{z\}$. The number of
edges in $F_1$ is also bounded by the number of backward edges from $V(T_2) \setminus \{z \}$ to $V(T_1)$. Since those forward/backward edges
are deleted by the partition, the actual increase is bounded by $|F_2|$, which satisfies $|F_2| = \tilde{O}(n/\eps_1)$ 
(see Lemma~\ref{lma:constDP}).
Hence we have $\sum_{G' \in \mathcal{G}_h} |E(G')| = \tilde{O}(m + n/\eps_1)$. It implies that the total 
time spent at each recursion level is $\tilde{O}(m / \eps_1 + n / \eps^2_1)$, and thus the total running time over 
all recursive calls is $\tilde{O}(m / \eps_1 + n / \eps^2_1)$. The size of 
the whole oracle is $\tilde{O}(n / \eps_1)$ because at each recursion level the data structure of size $\tilde{O}(n / \eps_1)$ 
is created.

The implementation of $\Query(x, t)$ is presented in Algorithm~\ref{query}, which basically follows the recursion approach
explained in Section~\ref{sec:outline}. We use notations $G^i.\Query$, $G^i.\Fquery$, and $G^i.\Dquery$ for clarifying 
the graph to which
the algorithm queries. Since the query processing traverses one downward path in the recursion tree
unless $x$ and $t$ is separated by the partition. If $x \in V(G^{2i})$ and $t \in V(G^{2i+1})$ happens, the case 1 of 
Section~\ref{sec:outline} (or the trivial case of $x \not\in V(P_{i})$) applies, and then no further recursion is invoked. 
The value returned as the answer of $G^i.\Query(x, t)$ is the minimum of all query responses.
It consists of at most $O(\log n)$
calls of $\Fquery$ and $\Dquery$. Since each call takes 
$O(\log n \cdot \log (\eps^{-1}\log (nW)))$ time (see Lemmas~\ref{lma:constDP} and \ref{lma:constPVSDO} for the details), 
the query time for $\Query(x, t)$ is $O(\log^2 n \cdot \log (\eps^{-1}\log (nW)))$. In summary, we have the following 
three lemmas.

\begin{lemma} \label{sizeOracle}
    The construction time of the $(1 + \eps)$-VSDO for $G$ is 
    $O(\eps^{-1} \log^4 n \cdot \log (nW) (m + n \eps^{-1} \cdot \log^3 n \cdot \log (nW)))$ 
    and the size of the constructed oracle is $O(\eps^{-1} n \log^3 n \cdot \log (nW))$.
\end{lemma}

\begin{lemma} \label{querytime}
    For any $x \in V(G) \setminus \{s \}$ and $t \in V(G)$, the running time of $G.\Query(x, t)$ is $O(\log^2 n \cdot \log (\eps^{-1}\log (nW)))$.
\end{lemma}

\begin{lemma} \label{lma:correctnessWholeoracle}
    For any $x \in V(G) \setminus \{s \}$ and $t \in V(G)$, $\Dist_{G-x}(s, t) \leq \Query(x, t) \leq (1+\eps) \cdot \Dist_{G-x}(s, t)$ holds.
\end{lemma}

The three lemmas above obviously deduce Theorem~\ref{thm:main}.

\begin{algorithm}[t]
    \caption{$G^i$.\textbf{Query$(x, t)$}}
    \label{query}
    \begin{algorithmic}[1]
        \If {$x$ is not an ancestor of $t$ in $T_i$}
            \State \Return $\Dist_{T_i}(s, t)$
        \EndIf
    
        \If {$|V(G^i)| \leq 6$}
            \State \Return $\Dist_{G^i-x}(s, t)$
        \EndIf
        
        \If {$x \in V(P_{i})$ and $t \in V(G^{2i+1})$}
            \State \Return $G^i.\Fquery(x, t)$ 
        \EndIf
        \If {$x \in V(G^{2i})$ and $t \in V(G^{2i})$}
            \State \Return $\min (G^{2i}.\Query(x, t), G^i.\Dquery(x, t))$
        \EndIf
        \If {$x \in V(G^{2i+1})$ and $t \in V(G^{2i+1})$}
            \State \Return $G^{2i+1}.\Query(v_f, t)$
        \EndIf
    \end{algorithmic}
\end{algorithm}

\section{Concluding Remarks}
We presented an algorithm which constructs a $(1+\eps)$-VSDO for directed weighted graphs.
The constructed oracle attains $\tilde{O}(n \log (nW) /\eps)$ size and $\tilde{O}(\log (nW) /\eps)$ query processing time. 
The construction time is $\tilde{O}(m \log (nW)/ \eps + n \log^2 (nW)/ \eps^2)$. We conclude the paper with a few open 
questions listed below:
\begin{itemize}
    \item Can we shave off the extra log factors of our oracle in size and query time, to match them with the best known bounds by~\cite{BCHR20}? 
    \item Can we extend our result for handling the \emph{dual failure model} admitting two edge/vertex failures? Very recently, the extension 
    of Bernstein's algorithm~\cite{Ber10} to the dual failure model has been proposed~\cite{CZ24}, which attains $\tilde{O}(n^2)$ running time. 
    It is an intriguing open question if there exists an $(1 + \eps)$-approximate SSRP algorithm for dual failure model running in $\tilde{O}(n^2)$ 
    time or not. 
    \item Is it possible to break the known conditional lower bounds in the case of \emph{additive approximation}?  
    \item Can all-pairs and multi-source cases benefit from the relaxation to $(1 + \eps)$-approximation?
\end{itemize}

\appendix

\section{Construction of \texorpdfstring{$P_T$}{PT}-faulty \texorpdfstring{$(1+\eps_1)$}{1-e}-VSDO} \label{sec:VSDO}
This section provides the details of our $P_T$-faulty $(1+\eps_1)$-VSDO construction, which includes the construction of 
the DP-Oracle.

\subsection{Construction of DP-Oracle}
Assume that $p$ is a power of $2$ for ease of presentation. To focus on the high-level idea, the detailed proofs of all the lemmas are 
defferred to the next section.
Let $\Pi(v_k, t)$ be the set of all $s$-$t$ departing paths avoiding $v_k$, and $\eps_2 = \eps_1 / (2 \log n) = O(\eps / \log^2 n)$.
We define $\Pi(t) = \bigcup_{k \in [0, p-1]} \Pi(v_k, t)$, and introduce the notations $\Gamma(I, t)$, 
$I^i_j$, and $\mathcal{I}$ as defined in Section~\ref{sec:outline}. We firs presenta a key technical lemma:

\begin{lemma} \label{lma:dporacle}
    Let $\Phi(i, j, t)$ be the predicate defined as follows.
    \begin{align*}
        \Phi(i, j, t) &\Leftrightarrow \text{$\left(j = 0 \right) \vee \left((1 + \eps_2) \cdot \Minlength(\Gamma(I^i_j, t)) < \min_{0 \leq j' < j} \Minlength(\Gamma(I^i_{j'}, t))\right)$}
    \end{align*}
    Assume a function $r : \mathbb{N} \times \mathbb{N} \times V(G) \to \mathbb{R}_+$ satisfying the following two conditions
    for any $t \in V(G) \setminus \{s\}$ and $0 \leq i \leq \log p$.
    \begin{enumerate}
        \item $\min_{0 \leq j' \leq j} \Minlength(\Gamma(I^i_{j'}, t)) \leq r(i, j, t) \leq r(i, j-1, t)$ for $j \geq 1$.
        \item If $\Phi(i, j, t)$ is true, $r(i, j, t) = \Minlength(\Gamma(I^i_j, t))$.
    \end{enumerate}
    Let $j(i, v_f)$ be the value satisfying $v_f \in I^i_{j(i, v_f)}$. Define $r'(v_f, t)$ as follows.
    \begin{equation*}
        r'(v_f, t) = \min_{0 \leq i \leq \log p} r(i, j(i, v_f)-1, t).
    \end{equation*}
    Then $\Minlength(\Pi(v_f, t)) \leq r'(v_f, t) \leq (1 + \eps_1) \cdot \Minlength(\Pi(v_f, t))$ holds for any $v_f \in V(P_T)$
    and $t \in V(G)$.
\end{lemma}

The function $r'$ in Lemma~\ref{lma:dporacle} obviously works as a $P_T$-faulty $(1 + \eps_1)$-DPO, and 
the value $r'(v_f, t)$ is easily computed from the function $r$. Hence the remaining issue is to construct 
the data structure of returning the value of $r(i, j, t)$ for $0 \leq i \leq \log p$, $0 \leq j < 2^i$, and $t \in V(G)$. 
The starting point is an algorithm of computing $\Minlength(\Gamma(I^i_j, v))$ for all $v \in V(G)$. Consider the graph $G_{i, j}$ obtained from $G$ by the following operations (see Fig. \ref{proggraph}):
\begin{enumerate}
    \item Remove all the edges in $P_T$, all incoming edges of vertices in $V(I^i_0) \cup V(I^i_1) \cup \dots \cup V(I^i_j)$, 
    and all outgoing edges of vertices in $V(I^i_{j+1}) \cup V(I^i_{j+2}) \cup \dots \cup V(I^i_{2^i-1})$.
    \item Add a new source vertex $s'$ and edges $(s', u)$ of weights $\Dist_{P_T}(s, u)$ for each $u \in I^i_j$.
\end{enumerate}

\begin{figure*}[t]
    \centering
    \includegraphics[keepaspectratio, scale=0.35]{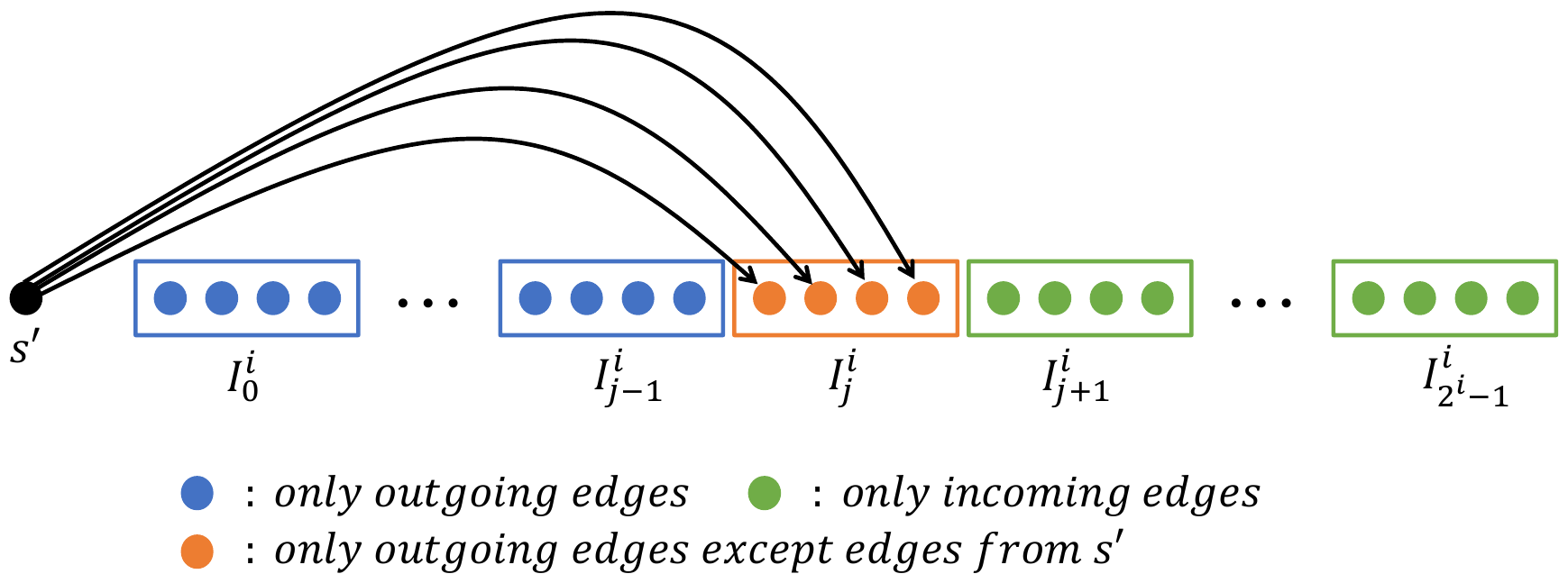}
    \caption{Example of $G_{i, j}$}
    \label{proggraph}
\end{figure*}

It is obvious that $\Dist_{G_{i, j}}(s', v) = \Minlength(\Gamma(I^i_{j'}, v))$ holds for any $v \in V(G)$. 
If we define $r$ as $r(i, j, v) = \min_{0 \leq j' < j} \Minlength(\Gamma(I^i_{j'}, v))$, the condition of Lemma~\ref{lma:dporacle} is 
satisfied. 
Hence running Dijkstra in $G_{i, j}$ for all $i$ and $j$ and storing all computation results provides the implementation 
of the data structure of $r$. However, the total running time of this approach is expensive, and thus not applicable. 
For any fixed $i$, the progressive Dijkstra computes the value $r(i, j, v)$ for all $j$ and $v$ approximately (in the 
sense of Lemma~\ref{lma:dporacle}) in $\tilde{O}(m/\eps)$ time. Roughly, the progressive Dijkstra iteratively 
applies a slightly modified version of the standard Dijkstra to $G_{i, j}$ in the increasing order of $j$. 
The modified points are summarized as follows:
\begin{itemize}
\item When processing $G_{i, 0}$, the distance vector $d$ managed by the Dijkstra algorithm 
initially stores $\infty$ for all vertices in $V(G)$. When processing $G_{i, j}$ for $j > 0$, the values $d[u]$ for 
$u \in V(I^i_j)$ are reset to $\infty$. For all other $u$, $d[u]$ keeps the results of processing $G_{i, j-1}$. 
\item All the vertices $u \in V(G) \setminus \{s'\}$ are inactive
at the beginning of processing $G_{i, j}$. It becomes active if the algorithm finds 
a $s$-$u$ path of length at most $d[u]/(1+ \eps_2)$ (i.e., the algorithm finds a path in $G_{i, j}$ whose length is substantially improved from the results for $G_{i, 0}, G_{i, 1}, \dots, G_{i, j-1}$). The vertex $u$ is added to the priority queue of the Dijkstra only 
when it becomes active. After becoming active, the update of $d[u]$ completely follows
the standard Dijkstra, i.e., it is updated if the algorithm finds a $s'$-$u$ path of length less than $d[u]$.
\item When $d[u]$ is updated by finding a shorter $s'$-$u$ departing path $Q$, the algorithm stores its branching vertex $b(Q)$
into the entry $b[u]$ of vector $b$.
\end{itemize}
The whole algorithm runs the above procedure for all $0 \leq i < \log p$. The pseudocode 
of the progressive Dijkstra is given in Algorithms~\ref{phase} and \ref{progdijk}.
The value $d[v]$ at the end of processing $G_{i, j}$ is used as the value of $r(i, j, v)$. Since it is costly to store 
those values explicitly, our algorithm stores the value of $r(i, j, v)$ only when $r(i, j - 1, v) > r(i, j, v)$ holds. If it holds, 
$v$ becomes active at processing $G_{i, j}$, and thus $r(i, j - 1, v)/(1 + \eps_2) > r(i, j, v)$ necessary holds. 
It implies that $r(i, 0, v), r(i, 1, v), \dots, r(i, 2^i-1, v)$ store at most $\left \lceil \log_{1+\eps_2}(nW) \right \rceil$ different 
values. Those values are stored in the list $\mathsf{upd}(i, v)$. The entry $(j, \ell, b) \in \mathsf{upd}(i, v)$ means that 
$r(i, j- 1, v) > r(i, j, v)$ and $r(i, j, v) = \ell$ hold, and the branching vertex of the corresponding $s'$-$v$ path is $b$.
Note that the information of the branching vertex is not necessary for constructing the DP-oracle, but required for the construction
of $G_1$. The time cost for accessing the value of $r(i, j, v)$ is $O(\log\log_{1+\eps_2}(nW))$, which is 
attained by a binary search over $\mathsf{upd}(i, v)$. We present the lemma claiming the correctness of this algorithm.

\begin{lemma} \label{dp-correctness}
    For any vertex $v \in V(G)$, $0 \leq i \leq \log p$ and $0 \leq j < 2^i$, let us define $r(i, j, v)$ as the value of $d[v]$ 
    after processing $G_{i, j}$. Then, $r$ satisfies the conditions (i), (ii) of Lemma~\ref{lma:dporacle}.
\end{lemma}

The running time of the progressive Dijkstra is bounded by the following lemma.
\begin{lemma} \label{runningtimeProgDijk}
The progressive Dijkstra processes $G_{i, j}$ for all $0 \leq j < 2^i$ 
and $0 \leq i < \log p$ in
$O(\eps_2^{-1} \cdot \log n \cdot \log (nW) \cdot (m + n\log n))$ time.
\end{lemma}
The actual entity of the $(1 + \epsilon_1)$-DPO we construct is $\mathsf{upd}$.
For all $i$ and $v$, $\mathsf{upd}(i, v)$ stores at most 
$\left \lceil \log_{1+\eps_2}(nW) \right \rceil$ different values. Hence the total size 
of the oracle is $O(\eps_2^{-1} n \cdot \log n \cdot \log (nW))$.  The computation of $r'$ (i.e., processing queries) takes $O(\log p)$ times of accessing $\mathsf{upd}$.
Each access takes $O(\log (\eps_2^{-1} \log (nW)))$ time, and thus the query time is $O(\log n \cdot \log (\eps_2^{-1} \log (nW)))$. 
By the fact of $\eps_2 = O(\eps_1 / \log n) = O(\eps / \log^2 n)$, the following lemma is obtained. 

\begin{lemma} \label{lma:constDP}
There exists an algorithm of constructing a $P_T$-faulty $(1 + \eps_1)$-DPO of size $O(\eps_1^{-1} n \cdot \log^2 n \cdot \log (nW))$. 
The construction time is $O(\eps_1^{-1} \cdot \log^2 n \cdot \log (nW) \cdot (m + n\log n))$ and the query processing time is $O(\log n \cdot \log (\eps_1^{-1} \log (nW)))$.
\end{lemma}

We also present the following auxiliary corollary, which is used in the construction of $G_1$.

\begin{corollary} \label{corol:dporacle}
For any $(\cdot, \ell, b) \in \mathsf{upd}(i, v)$, there exists a $s$-$v$ departing path of length $\ell$
with branching vertex $b$ in $G$. In addition, for any $v_f \in V(P_T)$ and $s$-$v$ shortest departing path 
$Q$ avoiding $v_f$ in $G$, there exists $0 \leq i \leq \log p$ and an entry $(\cdot, \ell, b(Q)) \in \mathsf{upd}(i, v)$ 
such that
$w(Q) \leq \ell \leq (1 + \eps_1) \cdot w(Q)$ holds.
\end{corollary}

\begin{algorithm}[t]
    \caption{\textbf{Prog-Main$(G, s, P_T, i)$}}
    \label{phase}
    \begin{algorithmic}[1]
        \State Create vectors $d$ and $b$ initialized with $d[v], b[v] \gets \infty$ for all $v \in V(G)$
        \For {$0 \leq j < 2^i$}
            \State create a $G_{i, j}$ 
            \State $d[v] \gets \infty$ for all $v \in V(I^i_j)$
            \State $b[v] \gets v$ for all $v \in V(I^i_j)$
            \State $\textbf{ProgDijk-Recursion}(G_{i, j}, s', d, b, i, j)$
        \EndFor
    \end{algorithmic}
\end{algorithm}

\begin{algorithm}[t]
    \caption{\textbf{ProgDijk-Recursion$(G, s, d, b, i, j)$}}
    \label{progdijk}
    \begin{algorithmic}[1]
        \State $d[s] \gets 0$
        \State Create an empty Fibonacci Heap $H$
        \State $\textbf{Insert}(H, s, 0)$
        \While {$H$ is not empty}
            \State $u \gets \textbf{Extract-Min}(H)$
            \For {$v \in N_G(u)$}
                \If {$d[u]+w_G(u, v) < d[v]$ and $v \in H$}
                    \State $d[v] \gets d[u]+w_G(u, v)$
                    \If {$u \neq s$}
                        \State $b[v] \gets b[u]$
                    \EndIf
                    \State $\textbf{Decrease-Key}(H, v, d[v])$
                    \State change $(j, \cdot, \cdot) \in \mathsf{upd}(i, v)$ to $(j, d[v], b[v])$
                \EndIf
                \If {$d[u]+w_G(u, v) < d[v]/(1+\eps_2)$ and $v \notin H$}
                    \State $d[v] \gets d[u]+w_G(u, v)$
                    \If {$u \neq s$}
                        \State $b[v] \gets b[u]$
                    \EndIf
                    \State $\textbf{Insert}(H, v, d[v])$
                    \State add $(j, d[v], b[v])$ to $\mathsf{upd}(i, v)$
                \EndIf
            \EndFor
        \EndWhile
    \end{algorithmic}
\end{algorithm}

\subsection{Correctness}

We provide the formal proofs of the lemmas in the previous section. To prove Lemma~\ref{lma:dporacle}, we first introduce an auxiliary lemma. 
For any sub-path $I$ of $P_T$, a set of sub-paths 
$C = \{J_0, J_1, \dots, J_{k-1}\} \subseteq \mathcal{I}$ is called a \emph{partition} of $I$ if 
$V(J_0), V(J_1), \dots, V(J_{k-1})$ exactly cover $V(I)$. 

\begin{lemma} \label{lma:exactcover}
    For any $b \in [0, f-1]$, there exists a partition $C$ of $P_T[0, b-1]$ such that $|C| \leq \lceil \log b \rceil$.
\end{lemma}

\begin{proof}
    We construct a partition $C$ of $P_T[0, b-1]$ by the following procedure.
    \begin{enumerate}
        \item $a \gets 0$.
        \item Find the largest $i$ such that $a + 2^i < b$ and $P_T[a, a+2^i-1] \in \mathcal{I}$ holds.
        \item Add $P_T[a, a+2^i-1]$ to $C$.
        \item $a \gets a + 2^i$.
        \item Return to the step 2 unless $a = b - 1$.
    \end{enumerate}
    If $a < b - 1$ holds in step 2, $i=0$ satisfies the condition. Hence the procedure eventually satisfies $a = b - 1$
    and terminates. The resulting $C$ is obviously a partition of $P_T[0, b-1]$. Let $a(j)$ be the value of $a$ at the beginning of the $j$-th iteration of the above process, and $i(j)$ be the value of $i$ at the second step of the $j$-th iteration.
    The inequality $b-a(j) \leq 2^{i(j)+1}$ holds for any $j$. Since $b-a(j+1) =b - a(j) - 2^{i(j)} \leq (b-a(j))/2$, the number of iterations is at most $\lceil \log b \rceil$.
    That is, $|C| \leq \lceil \log b \rceil$ holds.
\end{proof}

\begin{rlemma}{lma:dporacle}
    Let $\Phi(i, j, t)$ be the predicate defined as follows.
    \begin{align*}
        \Phi(i, j, t) &\Leftrightarrow \text{$\left(j = 0 \right) \vee \left((1 + \eps_2) \cdot \Minlength(\Gamma(I^i_j, t)) < \min_{0 \leq j' < j} \Minlength(\Gamma(I^i_{j'}, t))\right)$}
    \end{align*}
    Assume a function $r : \mathbb{N} \times \mathbb{N} \times V(G) \to \mathbb{R}_+$ satisfying the following two conditions
    for any $t \in V(G) \setminus \{s\}$ and $0 \leq i \leq \log p$.
    \begin{enumerate}
        \item $\min_{0 \leq j' \leq j} \Minlength(\Gamma(I^i_{j'}, t)) \leq r(i, j, t) \leq r(i, j-1, t)$ for $j \geq 1$.
        \item If $\Phi(i, j, t)$ is true, $r(i, j, t) = \Minlength(\Gamma(I^i_j, t))$.
    \end{enumerate}
    Let $j(i, v_f)$ be the value satisfying $v_f \in I^i_{j(i, v_f)}$. Define $r'(v_f, t)$ as follows.
    \begin{equation*}
        r'(v_f, t) = \min_{0 \leq i \leq \log p} r(i, j(i, v_f)-1, t).
    \end{equation*}
    Then $\Minlength(\Pi(v_f, t)) \leq r'(v_f, t) \leq (1 + \eps_1) \cdot \Minlength(\Pi(v_f, t))$ holds for any $v_f \in V(P_T)$
    and $t \in V(G)$.
\end{rlemma}

\begin{proof} Fix an arbitrary $t \in V(G) \setminus \{s\}$ throughout the proof. By the condition 1, we have
    \begin{equation} \label{eq1}
        r'(v_f, t) \geq \min_{0 \leq i \leq \log p} \min_{0 \leq j' < j(i, v_f)} \Minlength(\Gamma(I^i_{j'}, t)).
    \end{equation}
    Since $\Gamma(I^i_j, t) \subseteq \Pi(v_f, t)$ holds for any $0 \leq i \leq \log p$ and $0 \leq j < j(i, v_f)$,
    we also have
    \begin{equation} \label{eq2}
        \min_{0 \leq i \leq \log p} \min_{0 \leq j < j(i, v_f)}  \Minlength(\Gamma(I^i_j, t)) \geq \Minlength(\Pi(v_f, t)).
    \end{equation}
    Combining two expressions (\ref{eq1}) and (\ref{eq2}), we obtain $r'(v_f, t) \geq \Minlength(\Pi(v_f, t))$.
    Next, we prove $r'(v_f, t) \leq (1 + \eps_1) \cdot \Minlength(\Pi(v_f, t))$.
    If $\Phi(i, j, t)$ is true, we say that $(i, j)$ is \emph{strict}.
    Let $Q = \Minpath(\Pi(v_f, t))$ and $v_b = b(Q)$.
    From Lemma \ref{lma:exactcover}, there exists a partition $C = \{J_0, J_1, \dots, J_{|C|-1}\}$ of $P_T[0, b-1]$ with size 
    at most $\lceil \log b \rceil \leq \log p$.  For any $0 \leq k < |C|$, let $U_k$ be the prefix 
    $V(J_0) \cup V(J_1) \cup \dots,V(J_k)$ of $P_T$, and $i_k$ and $j_k$ be the values satisfying 
    $J_k = I^{i_k}_{j_k}$. We also define $h$ as the largest index such that $(i_h, j_h)$ is strict. Note that such $h$ 
    necessary exists because $(i_0, j_0)$ is always strict. By the condition 2, we have,
    \begin{equation} \label{eq3}
        r(i_h, j_h, t) = \Minlength(\Gamma(I^{i_h}_{j_h}, t)) \leq \Minlength(\Gamma(U_h, t)),
    \end{equation}
    where the right-side inequality comes from the fact that $\Phi(i_h, j_h, t)$ is true.
    Since $\Phi(i_{h'}, j_{h'}, t)$ is false for any $h' > h$, we also have 
    \begin{equation} \label{eq4}
        (1 + \eps_2) \cdot \Minlength(\Gamma(I^{i_{h'}}_{j_{h'}}, t)) \geq \min_{0 \leq j' < j_{h'}} \Minlength(\Gamma(I^{i_{h'}}_{j'}, t)).
    \end{equation}
    Since $V(I^{i_{h'}}_{0}) \cup V(I^{i_{h'}}_{1}) \cup \dots \cup V(I^{i_{h'}}_{j_{h'}-1}) = V(U_{h' - 1})$ holds, we obtain
    \begin{equation} \label{eq5}
        \min_{0 \leq j' < j_{h'}} \Minlength(\Gamma(I^{i_{h'}}_{j'}, t)) = \Minlength(\Gamma(U_{h'-1}, t)).
    \end{equation}
    By (\ref{eq4}) and (\ref{eq5}), the following inequality is obtained.
    \begin{equation} \label{eq6}
        (1 + \eps_2) \cdot \Minlength(\Gamma(I^{i_{h'}}_{j_{h'}}, t)) \geq \Minlength(\Gamma(U_{h'-1}, t)).
    \end{equation}
    By the fact of $\Minlength(\Gamma(U_{h'}, t)) = \min\{\Minlength(\Gamma(U_{h'-1}, t)), \Minlength(\Gamma(I^{i_{h'}}_{j_{h'}}, t))\}$ and (\ref{eq6}), 
    we have 
    \begin{equation} \label{eq7}
        \Minlength(\Gamma(U_{h' - 1}, t)) \leq (1 + \eps_2) \cdot \Minlength(\Gamma(U_{h'}, t)).
    \end{equation}
    In addition, by the definition of $Q$, 
    \begin{equation} \label{eq8}
        \Minlength(\Gamma(U_{|C|-1}, t)) = w(Q)
    \end{equation}
    holds. By (\ref{eq3}), (\ref{eq7}), and (\ref{eq8}), we conclude
    \begin{align*}
        r(i_h, j_h, t)  &\leq \Minlength(\Gamma(U_{h}, t)) \\
                        &\leq (1+\eps_2) \cdot \Minlength(\Gamma(U_{h+1}, t)) \\
                        &\leq  \cdots \\
                        &\leq (1 + \eps_2)^{|C|-1-h} \cdot \Minlength(\Gamma(U_{|C|-1}, t)) \\
                        &\leq (1 + \eps_2)^{\log p} \cdot w(Q) \\
                        &\leq (1 + \eps_1) \cdot w(Q).
    \end{align*}
    Since $r(i_h, j(i_h, v_f)-1, t) \leq r(i_h, j_h, t)$ holds by the condition 1, we obtain $r'(v_f, t) \leq r(i_h, j_h, t)$.
    The lemma is proved.
\end{proof}

\begin{rlemma}{dp-correctness}
    For any vertex $v \in V(G)$, $0 \leq i \leq \log p$, and $0 \leq j < 2^i$, let us define $r(i, j, v)$ as the value of $d[v]$ 
    after processing $G_{i, j}$. Then, $r$ satisfies the conditions 1 and 2 of Lemma~\ref{lma:dporacle}.
\end{rlemma}

\begin{proof}
\sloppy{
    First, we consider the condition 1. 
    At the end of processing $G_{i,j}$, all the graphs $G_{i, 0}, G_{i, 1}, \dots, G_{i, j}$ 
    have been processed. Since $r(i, j, v)$ is not less than the length of the shortest $s'$-$v$ path in any $G_{i, j'}$ of $0 \leq j' \leq j$, 
    we have
    \begin{align*}
        r(i, j, v) \geq \min_{0 \leq j' \leq j} \Dist_{G_{i, j'}}(s', v) = \min_{0 \leq j' \leq j} \Minlength(\Gamma(I^i_{j'}, v)).
    \end{align*}
    In addition, $r$ is obviously non-increasing with respect to $j$. Hence the condition 1 is satisfied.
    Consider the condition 2. For $j = 0$, the standard Dijkstra is executed to process $G_{i,0}$. Hence the condition 
    2 obviously holds.
    For $j \geq 1$, $\Phi(i, j, v)$ implies $\Minlength(\Gamma(I^i_j, v)) = \min_{0 \leq j' \leq j} \Minlength(\Gamma(I^i_{j'}, v))$.
    Since $r(i, j, v) \geq \min_{0 \leq j' \leq j} \Minlength(\Gamma(I^i_{j'}, v))$ holds by the condition 1, 
    $\Phi(i, j, v) \Rightarrow (r(i, j, v) = \Minlength(\Gamma(I^i_j, v)))$ and $\Phi(i, j, v) \Rightarrow (r(i, j, v) \leq \Minlength(\Gamma(I^i_j, v)))$ are equivalent. Hence what we show is 
    \begin{align} \label{eq:keypred}
        r(i, j, v) > \mathsf{minL}(\Gamma(I^i_j, v)) \Rightarrow
        (1+\eps_2) \cdot \mathsf{minL}(\Gamma(I^i_j, v)) \geq \min_{0 \leq j' < j} \mathsf{minL}(\Gamma(I^i_{j'}, v)),
    \end{align}
    which deduces the condition 2. Let $R_{i, j, v} = \Path{a_0, a_1, \dots, a_{k-1}}$ $(s'=a_, v=a_{k-1})$ be the shortest path from $s'$ to $v$ in $G_{i, j}$. 
    By the precondition of the predicate above, $r(i, j, v)$ is not equal to $w(R_{i,j, v})$. It implies that there exists a vertex 
    $a_\ell$ on $R_{i, j, v}$ which does not become active in processing $G_{i,j}$, because $d[v]$ is correctly 
    updated with the value $w(R_{i,j, v})$ if $a_0, a_1, \dots, a_{k-1}$ are all active. 
    Without loss of generality, we assume that $a_{\ell}$ is such a vertex with the smallest $\ell$. Since all the vertices 
    $u \in V(I^i_j)$ are necessarily active due to the reset of $d[u]$, $a_{\ell}$ is a
    vertex not on $P_T$.  At the end of processing $G_{i, j}$, $(1 + \eps_2) \cdot w(R_{i,j, a_\ell}) > d[a_\ell] = r(i, j, a_{\ell})$ holds.
    By definition, $w(R_{i, j, a_\ell}) = \Dist_{G_{i,j}}(s', a_\ell) = \Minlength(\Gamma(I^i_{j}, a_\ell))$ holds, and thus
    we conclude $r(i, j, a_{\ell}) < (1 + \eps_2) \cdot \Minlength(\Gamma(I^i_{j}, a_\ell))$.
    Let $j' < j$ be the largest value such that $a_\ell$ becomes active in processing $G_{i, j'}$. 
    Then we have $r(i, j, a_\ell) = r(i, j', a_\ell) = w(R_{i, j', a_\ell})$. We consider a new path $Q$ obtained by concatenating $R_{i, j', a_{\ell}}$
    and $\Path{a_{\ell}, a_{\ell+1}, \dots, a_{k-1}}$.
    The length of $Q$ is bounded as follows:
    \begin{align*}
        w(Q) &= r(i, j', a_\ell) + w(\Path{a_{\ell}, a_{\ell+1}, \dots, a_{k-1}}) \\
             &= r(i, j, a_\ell) + w(\Path{a_{\ell}, a_{\ell+1}, \dots, a_{k-1}}) \\ 
             &\leq (1+\eps_2) \cdot \mathsf{minL}(\Gamma(I^i_j, a_{\ell})) + w(\Path{a_{\ell}, a_{\ell+1}, \dots, a_{k-1}}) \\
             &\leq (1+\eps_2) \cdot \mathsf{minL}(\Gamma(I^i_j, t)).
    \end{align*}
    It is easy to check that $Q$ also exists in $G_{i, j'}$. Hence we have $\min_{0 \leq j' < j} \mathsf{minL}(\Gamma(I^i_{j'}, v)) \leq 
    w(R_{i,j, v}) \leq w(Q)$, i.e.,  the consequence side of the predicate~(\ref{eq:keypred}) holds. The lemma is proved. 
    }
\end{proof}

\begin{rlemma}{runningtimeProgDijk}
The progressive Dijkstra processes $G_{i, j}$ for all $0 \leq j < 2^i$ 
and $0 \leq i < \log p$ in
$O(\eps_2^{-1} \cdot \log n \cdot \log (nW) \cdot (m + n\log n))$ time.
\end{rlemma}

\begin{proof}
    We first bound the time for constructing $G_{i, j}$. Initially, $G_{i, 0}$ is constructed in $O(m)$ time. The graph $G_{i, j}$ is
    obtained by modifying $G_{i, j-1}$, which takes $O(\sum_{u \in V(I^i_j)} |N_G(u)|)$ time. Summing up it for all $j$, the total 
    construction time is bounded by $O(m)$. Next, we count the number of invocations of \textbf{Extract-Min}, \textbf{Insert} and 
    \textbf{Decrease-Key} respectively.
    When each vertex $u$ not on $P_T$ is inserted to the heap, $d[u]$ decreases to its $1/(1+\eps_2)$ fraction. Hence
    the total number of the invocations of \textbf{Insert} for vertex $u$ is at most $\left \lceil \log_{1+\eps_2} (nW) \right \rceil$. 
    For vertex $u \in I^i_j$, it is inserted to the heap at most $\left \lceil \log_{1+\eps_2} (nW) \right \rceil$ times in the execution
    up to processing $G_{i, j}$, and in the following execution it is never inserted. Consequently, the total number of invoking
    \textbf{Insert} and \textbf{Extract-Min} is $O(n \left \lceil \log_{1+\eps_2} (nW) \right \rceil)$. 
    Next, we count the number of invocations of \textbf{Decrease-Key}.
    For each vertex $u$, \textbf{Decrease-key} is called when its adjacent vertex of $v$ is extracted from the heap.
    That is, for one insertion of $v$ into the heap, at most $|N_G(v)|$ invocations of \textbf{Decrease-key} follows.
    Since each vertex is inserted to the heap at most $\left \lceil \log_{1+\eps_2} (nW) \right \rceil$ times, the total number of calling \textbf{Decrease-Key} is $O(m \left \lceil \log_{1+\eps_2} (nW) \right \rceil)$. The amortized operation cost of Fibonacci heap is $O(1)$ for \textbf{Insert} and \textbf{Decrease-key} and $O(\log n)$ for
    \textbf{Extract-min}, and $\log_{1+\eps_2} (nW) = O(\eps^{-1}_2 \log (nW))$ holds. 
    Consequently, the total running time of the progressive Dijkstra is $O(\eps_2^{-1} \cdot \log n \cdot \log (nW) \cdot (m + n\log n))$.
\end{proof}

\begin{rlemma}{lma:constDP}
There exists an algorithm of constructing a $P_T$-faulty $(1 + \eps_1)$-DPO of size $O(\eps_1^{-1} n \cdot \log^2 n \cdot \log (nW))$. 
The construction time is $O(\eps_1^{-1} \cdot \log^2 n \cdot \log (nW) \cdot (m + n\log n))$ and the query processing time is $O(\log n \cdot \log (\eps_1^{-1} \log (nW)))$.
\end{rlemma}

\begin{proof}
Obvious from the argument in Section~\ref{sec:outline} and Section~\ref{sec:VSDO} (recall $\eps_1 = O(\eps_2 / \log n)$).
\end{proof}

\subsection{Handling Jumping Paths}

The following lemma is the key fact of estimating $\JDist_{G - v_f}(s, t)$.

\begin{lemma} \label{lma:jumping-propery}
    For any failed vertex $v_f \in V(P_T)$ and destination $t \in V(T_2) \setminus \{z\}$, 
    $\JDist_{G-v_f}(s, t) = \Dist_{G-v_f}(s, z) + \Dist_T(z, t)$ holds 
    if $\DDist_{G-v_f}(s, t) \geq \JDist_{G-v_f}(s, t)$ is satisfied.
\end{lemma}

\begin{proof}
    Let $R$ be the shortest $s$-$t$ jumping path for $v_f$, and $Q$ be the shortest path from $c(R)$ to $t$ on $T$.
    Obviously, $t$ is a descendant of $c(R)$ because $t \in V(T_2)$ and $c(R) \in V(P_T)$ holds.
    Thus, $Q$ is also the shortest path from $c(R)$ to $t$ on $G-v_f$.
    Let $R'$ be the walk obtained by replacing the suffix of $R$ after $c(R)$ with $Q$, and $R''$ be the walk obtained 
    by replacing the prefix of $R'$ before $z$ with the $s$-$z$ shortest path $Q'$ in $G-v_f$.
    By the shortest property of $Q$ and $Q'$, $w(R'') \leq w(R') \leq w(R)$ holds.
    If $R''$ is not a simple path, the suffix of $Q$ from $z$ intersects $Q'$. 
    Then, there is a departing path avoiding $v_f$ in $R'$ with a distance shorter than $R$. That is,
    $\DDist_{G-v_f}(s, t) \geq \JDist_{G-v_f}(s, t)$ holds, which contradicts the assumption of the lemma.
    Hence $R''$ is a simple path and a $s$-$t$ jumping path for $v_f$.
    By the shortest property of $R$, $w(R) = w(R'')$ holds. It implies 
    $\JDist_{G-v_f}(s, t) = w(R) = w(R'') = \Dist_{G-v_f}(s, z) + \Dist_T(z, t)$.
\end{proof}

By the lemma above, the length of the shortest $s$-$t$ jumping path is estimated by storing the values of 
$\Dist_{G - v_f}(s, z)$ for all $v_f \in V(P_T)$. 
Those values are computed approximately by the algorithm by Bernstein~\cite{Ber10},
which is stated as follows: 
\begin{theorem}[Bernstein \cite{Ber10}] \label{thm:(1+e)-RP}
Given any directed graphs $G$ with edge weights in range $[1, W]$ and two vertices $s, t \in V(G)$, there exists an algorithm that calculate the value $\hat{d}(x)$ such that 
$\Dist_{G-x}(s, t) \leq \hat{d}(x) \leq (1+\eps_1) \cdot \Dist_{G-x}(s, t)$ for any failed 
vertex $x \in V(G)$ in $O(\eps_1^{-1} \cdot \log^2 n \cdot \log (nW) \cdot (m + n\log n))$ 
time.
\end{theorem}
The computed results are stored into an array of $O(n)$ size, and can be accessed in $O(1)$ time. 
Hence the query time for jumping paths is $O(1)$.
Combining the construction of the DP-oracle (Lemma~\ref{lma:constDP}), we obtain the following lemma.
\begin{lemma} \label{lma:constPVSDO}
There exists an algorithm of constructing a $P_T$-faulty $(1 + \eps_1)$-VSDO of size $O(\eps_1^{-1} n \cdot \log^2 n \cdot \log (nW))$. 
The construction time takes $O(\eps_1^{-1} \cdot \log^2 n \cdot \log (nW) \cdot (m 
+ n\log n))$ and the query processing time is $O(\log n \cdot (\eps_1^{-1} \log (nW)))$.
\end{lemma}

\begin{proof}
\sloppy{
The construction time and query processing time follow Lemmas~\ref{lma:constDP} and \ref{thm:(1+e)-RP}. Hence we focus on the 
correctness. For any query $(v_f, t)$, $P_T$-faulty $(1 + \eps_1)$-VSDO returns $\min(\Dquery(v_f, t), \hat{d}(v_f) +\Dist_T(z, t))$, where $\hat{d}(v_f)$ is the $(1 + \eps_1)$-approximate value of $\Dist_{G - v_f}(s, z)$ computed
by the algorithm of Theorem~\ref{thm:(1+e)-RP}. If $\DDist_{G-v_f}(s, t) \leq \JDist_{G-v_f}(s, t)$ holds, 
we obtain $\Dist_{G - v_f}(s, x) \leq \Dquery(v_f, t) \leq (1 + \eps_1) \Dist_{G-v_f}(s, x)$ by Lemma \ref{lma:constDP}. 
Since $\Dist_{G - v_f}(s, x) \leq \hat{d}(v_f) +\Dist_T(z, t))$ is obvious, the oracle correctly returns a 
$(1 + \eps_1)$-approximate value of $\Dist_{G - v_f}(s, t)$. In the case of $\DDist_{G-v_f}(s, t) > \JDist_{G-v_f}(s, t)$, 
Lemma~\ref{lma:jumping-propery} and Theorem~\ref{thm:(1+e)-RP} imply that $\hat{d}(v_f) + \Dist_T(z, t)$ is 
a $(1+\eps_1)$-approximation of $\Dist_{G-v_f}(s, x)$.
}
\end{proof}

\section{Proof of Lemma~\ref{lma:G2-correctness}}

\begin{rlemma}{lma:G2-correctness}
    For any failed vertex $x \in V(T_2) \setminus \{z\}$ and destination $t \in V(T_2) \setminus \{z\}$, 
    $\Dist_{G_2-x}(s, t) = \Dist_{G-x}(s, t)$ holds.
\end{rlemma}

\begin{proof}
To prove $\Dist_{G_2-x}(s, t) \leq \Dist_{G-x}(s, t)$,  
            we show that for any shortest $s$-$t$ path in $G-x$ there exists a $s$-$t$ path with the same length in $G_2-x$.
            Let $R_G = \Path{a_0, a_1, \cdots, a_{k-1}}$ $(s=a_0, t=a_{k-1})$ be any shortest $s$-$t$ path in $G-x$, and $a_i$ 
            be the vertex in $V(R_G) \cap V(T_2)$ with the largest $i$.
            Consider the path $R'_G$ obtained by concatenating the $s$-$a_{i-1}$ path $Q$ in $T$ and $\Path{a_{i-1}, a_i, \dots, a_{k-1}}$.
            It is obvious that $R'_G$ avoids $x$.
            Since $Q$ is the shortest, $w(R'_G)=w(R_G)=\Dist_{G-x}(s, t)$ holds.
            Due to the construction of $G_2$, there exists an edge from $s$ to $a_i$ of weights $\Dist_T(s, a_{i-1})+w_G(a_{i-1}, a_i)$.
            Thus there exists a path $R_{G_2} = \Path{s, a_i, a_{i+1}, \dots, a_{k-1}}$ avoiding $x$ in $G_2$, i.e., a path in $G_2 - x$.

Next we prove $\Dist_{G_2-x}(s, t) \geq \Dist_{G-x}(s, t)$. For the proof, 
            we show that for any $s$-$t$ path in $G_2-x$, there exists a $s$-$t$ path with the same length distance in $G-x$.
            Let $R_{G_2} = \Path{a_0, a_1, \dots, a_{k-1}}$ $(s=a_0, t=a_{k-1})$ be any $s$-$t$ path on $G_2-x$.
            Since $s$ and $a_1$ are connected by an edge, there exists a vertex $y \in V(T_1)$ such that $\Dist_T(s, y) + w(y, a_1)$.
            Consider the path $R_G$ obtained by concatenating the $s$-$y$ path in $T$ and $\Path{y, a_1, a_2, \dots, a_{k-1}}$.
            Obviously, $R_G$ is a path in $G-x$, and thus $w(R_G) = w(R_{G_2}) = \Dist_{G_2-x}(s, t)$ holds.
\end{proof}

\section{Proof of Lemma~\ref{lma:G1-correctness}}

For proving Lemma~\ref{lma:G1-correctness}, we introduce an auxiliary lemma below.

\begin{lemma} \label{lma:pisbetter}
For any $e = (v_b, v_c) \in F_2$, $\Dist_{P_T}(v_b, v_c) \leq w(e)$ holds. In addition, if $x \in V(P_T) \setminus \{v_b, v_c\}$ holds, $G - x$ contains a $v_b$-$v_c$ path of length $w(e)$.
\end{lemma}

\begin{proof}
By the construction of $F_2$ and Corollary~\ref{corol:dporacle}, there exists a $s$-$v_c$ 
departing path $Q$ of length $w(e) + \Dist_T(s, v_b)$ with branching vertex $v_b$ in $G$.
That is, the $v_b$-$v_c$ sub-path $Q'$ of $Q$ has length $w(e)$. Hence 
$\Dist_G(v_b, v_c) = \Dist_{P_T}(v_b, v_c) \leq w(e)$ holds. Since $Q'$ is the suffix from $b(Q)$, 
$V(Q') \setminus \{v_c\}$ does not intersect $V(P_T)$. If $x \in V(P_T)$, $Q'$ is also a path in $G - x$.
\end{proof}

The proof of Lemma~\ref{lma:G1-correctness} is given below:
\begin{rlemma}{lma:G1-correctness}
    For any failed vertex $x \in V(T_1)$ and destination $t \in V(T_1)$, the following conditions hold:
    \begin{itemize}
        \item If $x \not\in V(P_T)$, $\Dist_{G_1-x}(s, t) =  \Dist_{G-x}(s, t)$.
        \item If $x \in V(P_T)$ and $\JDist_{G -x}(s,t) < \DDist_{G - x}(s, t)$ holds, 
        $\Dist_{G - x}(s, t) \leq \Dist_{G_1-x}(s, t) \leq (1 + \eps_1) \cdot \Dist_{G -x}(s, t)$.
        \item Otherwise, $\Dist_{G-x}(s, t) \leq  \Dist_{G_1-x}(s, t)$.
    \end{itemize}
\end{rlemma}

\begin{proof}
    For proving the lemma, we associate each edge $e \in F_1 \cup F_2$ with the path $Q_e$ in $G$, which is defined
    as follows:
    \begin{itemize}
    \item If $e = (z, u) \in F_1$ holds, $Q_e$ is the shortest $z$-$u$ path in $G[\Iset(V(T_2))]$. 
    \item If $e = (v_b, v_c) \in F_2$ holds, $Q_e$ is the sub-path from $v_b$ to $v_c$ of the 
    $s$-$v_c$ departing path which is the shortest among the ones with branching vertex $v_b$ (if two or more shortest paths exist, chosen arbitrarily). Such a sub-path necessarily exists by Lemma~\ref{lma:pisbetter}.
    \end{itemize}
    Obviously, $w(e) = w(Q_e)$ holds for any $e \in F$. Let $R_{G_1}$ be the shortest $s$-$t$ path in $G_1 - x$.
    We first prove the lower bound side of all cases. Consider the case of $x \not\in V(P_T)$. If $R_{G_1}$ contains an edge 
    $e = (v_b, v_c)$ in $F_2$, one can remove this by replacing $e$ with the $v_b$-$v_c$ path in $P_T$ without increasing the length. Hence we assume that $R_{G_1}$ does not contain any edge in $F_2$. For any edge $e \in F_1$, $Q_e$ does not intersect $V(T_1) \setminus \{z\}$,
    and thus $Q_e$ does not contain $x$. Consequently, the path $R'_{G_1}$ obtained by replacing all $e \in F_1$ with $Q_e$ is a
    path in $G - x$, and satisfies $w(R'_{G_1}) = w(R_{G_1})$. It implies $\Dist_{G - x}(s, t) \leq \Dist_{G_1 - x}(s, t)$. Next,
    consider the case of $x \in V(P_T)$. By Lemma~\ref{lma:pisbetter}, one can assume that any path $Q_e$ for $e \in F_2$ does not contain 
    $x$, and thus it is a path in $G - x$. Hence by replacing each edge $e \in E(R_{G_1}) \cap F$ with $Q_e$, we obtain the path $R'_{G_1}$
    in $G - x$ with length $w(R'_{G_1}) = w(R_{G_1})$. It also implies $\Dist_{G - x}(s, t) \leq \Dist_{G_1 - x}(s, t)$.
    
    We prove the upper bound side of the first two cases. Let $R_G$ be the shortest $s$-$t$ path in $G - x$. If $R_G$ does not contain 
    any vertex in $V(T_2) \setminus \{z\}$, the conditions obviously hold, and thus 
    we assume $V(R_G) \cap (V(T_2) \setminus \{z\}) \neq \emptyset$.
    Consider the case of $x \not\in V(P_T)$. Let $R_G$ be the shortest $s$-$t$ path in 
    $G - x$. If $R_G$ does not contain any vertex in $V(T_2)$,  it is also a path in $G_1 - x$ and thus $\Dist_{G - x} \geq \Dist_{G_1 - x}(s, t)$ obviously holds. 
    Otherwise, let $(u, v)$ be the last backward edge from $V(T_2) \setminus \{z \}$ to $V(T_1)$ in $R_G$. 
    Since the $s$-$u$ path $Q$ in $T$ does not contain $x$, it is also the shortest $s$-$u$ path in $G - x$. 
    Hence one can assume that the prefix of $R_G$ up to $u$ is equal to $Q$ without loss of generality. 
    It implies that $R_G$ contains $z$, and for its sub-path from $z$ to $v$ there exists an edge $(z, v) \in F_1$ of weight $\Dist_T(z, u) + w_G(u, v)$. By replacing $Q$ in $R_G$ with
    $(z, v) \in F_1$, we obtain a $s$-$t$ path in $G_1 - x$ with length $w(R_G)$, i.e., we conclude 
    $\Dist_{G_1 - x}(s, t) \leq \Dist_{G -x}(s, t)$. Next consider the case of $x \in V(P_T)$ and $\JDist_{G -x}(s,t) < \DDist_{G - x}(s, t)$. Then $R_G$ is a jumping path avoiding $x$. Since $R_G$ contains a vertex in $V(T_2) \setminus \{z\}$, one can assume
    that $R_G$ contains $z$ and a backward edge from $V(T_2) \setminus \{z\}$ to $V(T_1)$ by the same argument as the proof of Lemma~\ref{lma:jumping-propery}. Let $(u, v)$ be the last backward edge in $R_G$. 
    Then the length of the sub-path $W_G$ from $z$ to $v$ of $R_G$ must be equal to $w(Q_e)$ for $e = (z, v) \in F_1$. Let $W'_G$ be the 
    prefix of $R_G$ up to $c(R_G)$. Since it is the shortest $s$-$c(R_G)$ departing path avoiding $x$,
    by Lemma~\ref{lma:pisbetter}, there exists $i$ and $(\cdot, \ell, b(R_G)) \in \mathsf{upd}(i, c(R_G))$ such that
    $w(W'_G) \leq \ell (1 + \eps_1) w(W'_G)$ holds. It implies that there exists an edge $e' = (b(R_G), c(R_G)) \in F_1)$ of weight 
    $\ell - \Dist_T(s, b(R_G))$. We obtain $R'_G$ by replacing its sub-paths $W_G$ and $W'_G$ with $e$ and $e'$ respectively.
    Then we have $w(R'_G) = \Dist_T(s, b(R_G)) + (\ell - \Dist_T(s, b(R_G))) + \Dist_T(c(R_G), z) + w(W_G) = \ell + \Dist_T(c(R_G), z) + w(W_G)$, which is bounded by
    $(1 + \eps_1) \cdot w(W'_G) + \Dist_T(c(R_G), z) + w(W_G) \leq (1 + \eps_1) \cdot w(R_G)$.
\end{proof}

\subsection{Omitted Proofs in Section~\ref{sec:constAndQuery}}

\begin{rlemma}{sizeOracle}
    The construction time of the $(1 + \eps)$-VSDO for $G$ is 
    $O(\eps^{-1} \log^4 n \cdot \log (nW) (m + n \eps^{-1} \cdot \log^3 n \cdot \log (nW)))$ 
    and the size of the constructed oracle is $O(\eps^{-1} n \log^3 n \cdot \log (nW))$.
\end{rlemma}

\begin{proof}
    Consider the construction of $G^{2i}$ and $G^{2i+1}$ from $G^i$. Let $T_i$, $T_{2i}$, and $T_{2i+1}$ be the shortest path trees 
    of $G^i$, $G^{2i}$, and $G^{2i+1}$, and $z_i$ be the centroid of $T_i$.
    Since $T_{2i}$ and $T_{2i+1}$ share the centroid $z_i$, $|V(G^{2i})| + |V(G^{2i+1})| = |V(G^i)|+1$ holds.
    Let $X_i$ be the set of the edges crossing between $G^i[V(T_{2i}]$ and $G^i[V(T_{2i+1}] \setminus \{z_i\}$.
    The total number of the edges added for constructing $G^{2i + 1}$ and the edges in $F_1$ for constructing $G^{2i}$ 
    is bounded by $|X_i|$. Thus the essential increase of the number of edges is caused by the addition of $F_2$ for constructing 
    $G^{2i}$, i.e., the sum of the cardinality of $\mathsf{upd}(i, v)$ for all $0 \leq i \leq \log p$ and $v \in V(G^i)$.
    It is bounded by $O(|V(G^i)| \cdot \left \lceil \log_{1+\eps_2} (nW) \right \rceil)$. Consequently, we have 
    $|E(G^{2i})|+|E(G^{2i+1})| \leq |E(G^i)| + |V(G^i)| \cdot \left \lceil \log_{1+\eps_2} (nW) \right \rceil$.

    We bound the total number of recursions for constructing the oracle. The centroid bipartition of $T_i$ divides their edges
    into two disjoint subsets. That is, we have $E(T_i) = E(T_{2i}) \cup E(T_{2i+1})$ and $E(T_{2i}) \cap E(T_{2i+ 1}) = \emptyset$.
    Then $O(n)$ recursive calls suffice to reach the decomposition of $G$ into the set of constant-size ($\leq 6$) subgraphs.
    Using the (in)equalities $|V(G_{2i})| + |V(G_{2i+1})| = |V(G^i)|+1$ and $|E(G^{2i})|+|E(G^{2i+1})| \leq |E(G^i)| + |V(G^i)| \cdot \left \lceil \log_{1+\eps_2} (nW) \right \rceil$, we obtain $\sum_{G' \in \mathcal{G}_h} |V(G')| \leq 2n-2$ and 
    $\sum_{G' \in \mathcal{G}_h} |E(G')| \leq m + (2n-2) \cdot h \cdot \lceil \log_{1+\eps_2} (nW) \rceil
    = O(m + n \eps^{-1}_1 \cdot \log^2 n \cdot \log (nW))$. Since it is obvious that the recursion depth is bounded by $O(\log n)$,
    the total construction time and size of the sub-oracles for all subgraphs in $\mathcal{G}_h$ are respectively 
    bounded by $O(\eps^{-1}_1 \log^3 n \cdot \log (nW) (m + n \eps^{-1}_1 \cdot \log^2 n \cdot \log (nW)))$ and $O(\eps^{-1}_1 n \log^2 n \cdot \log (nW))$. 
    By $\eps_1 = O(\eps / \log n)$, the total construction time and size over all depths are also bounded as the statement of the lemma. 
\end{proof}

\begin{rlemma}{querytime}
    For any $x \in V(G) \setminus \{s \}$ and $t \in V(G)$, the running time of $G.\Query(x, t)$ is $O(\log^2 n \cdot \log (\eps^{-1}\log (nW)))$.
\end{rlemma}

\begin{proof}
The time for query processing at each recursion level is dominated by the query to the DP-oracle (notice that the time for 
processing $\Fquery$ is also dominated by the query to the DP-oracle inside), which takes 
$O(\log n \cdot \log (\eps^{-1} \log (nW)))$ time. Since the query processing algorithm traverses a path in the recursion tree, 
the total number of queries to the DP-oracles is $O(\log n)$. The lemma is proved.
\end{proof}

\begin{rlemma}{lma:correctnessWholeoracle}
    For any query $(x, t)$ ($x$ is a failed vertex and $t$ is a destination vertex), let $\hat{d}_{G^i - x}(s, t)$ be the result of $G^i.\Query(x, t)$. Then, $\Dist_{G-x}(s, t) \leq \hat{d}_{G^1 -x}(s, t) \leq (1+\eps) \cdot \Dist_{G-x}(s, t)$ holds.
\end{rlemma}

\begin{proof}
    Let $G^{i_1}, G^{i_2}, \dots, G^{i_k}$ be the path in the recursion tree traversed in processing $G^1.\Query(x, t)$ 
    ($G^{i_1} = G^1$, and the recursion terminates at $G^{i_k}$).
    By Lemmas~\ref{lma:G2-correctness} and \ref{lma:G1-correctness}, $\Dist_{G^{i_j}-x}(s, t) \leq \Dist_{G^{i_{j+1}} - x}(s, t)$ always holds
    for any $1 \leq j < k$. Then it is easily shown that the values $\hat{d}_{G^{i_j} -x}(s, t)$ returned by each $G^{i_j}.\Query(x, t)$ is lower bounded by $\Dist_{G^i_j - x}(s, t)$. Hence it suffices to show 
    $\hat{d}_{G^{i_j} -x}(s, t) \leq (1 + \eps) \Dist_{G^{i_j} - x}(s, t)$ holds for some $1 \leq j \leq k$. Let 
    $j'$ be the smallest value satisfying either one of them.
    \begin{itemize}
    \item $j' = k$ holds.
    \item $\DDist_{G^{i_{j' + 1}} - x}(s, t) \leq \JDist_{G^{i_{j'+ 1}} - x}(s, t)$ holds.
    \end{itemize}
    \sloppy{
    First, consider the case of $j' = k$. There are three possible sub-cases:
    \begin{itemize}
    \item $x \in V(P_{T, i_{j'}})$ and $t \in V(G^{2i_{j'} +1}) \setminus \{z_{i_{j'}} \}$. 
    \item $x \not\in V(P_{i_{j'}})$ and $V(G^{2i_{j'} +1}) \setminus \{z_{i_{j'}} \}$.
    \item $|v(G^{i_{j'}})| \leq 6$.
    \end{itemize}
    In the first sub-case, the value returned by 
    $G^{i_{j'}}.\Fquery(x, t)$ is a $(1 + \eps_1)$-approximation of $\Dist_{G^{i_{j'}} - x}(s, t)$. 
    In the second and third sub-cases, $G^{i_{j'}}.\Query(x, t)$ obviously returns $\Dist_{G^{i_{j'}} - x}(s, t)$. 
    Next, consider the case of $\DDist_{G^{i_{j' + 1}} - x}(s, t) \leq \JDist_{G^{i_{j'+ 1}} - x}(s, t)$.
    Then $G^{i_{j'}}.\Dquery(x, t)$ returns 
    a $(1 + \eps_1)$-approximation of $\DDist_{G^{i_{j'}} - x}(s, t) = \min\{\DDist_{G^{i_{j'}} - x}(s, t), \JDist_{G^{i_{j'}} - x}(s, t) \} = \Dist_{G^{i_{j'}} - x}(s, t)$ by Lemma 
    \ref{lma:constDP}. In any case, we have 
    $\hat{d}_{G^{i_{j'}}-x}(s, t) \leq (1+\eps_1) \cdot \Dist_{G^{i_{j'}}-x}(s, t)$.
    By Lemmas~\ref{lma:G2-correctness} and \ref{lma:G1-correctness}, we have 
    $\Dist_{G^{i_{j}} - x} \leq (1 + \eps_1) \cdot \Dist_{G^{i_{j - 1}} - x}(s, t)$ for any $1 < j \leq j'$. That is,
    we have $\Dist_{G^{i_{j'}} - x} \leq (1 + \eps_1)^{j' - 1} \Dist_{G^{1} - x}(s, t)$. 
    It implies $\hat{d}_{G^{i_{j'}}-x}(s, t) \leq (1+\eps_1)^{j'} \cdot \Dist_{G^{1}-x}(s, t)$. Since $j' \leq 2\log n$,
    $(1 + \eps_1)^{j'} \leq (1 + \eps)$ holds. The lemma is proved.}
\end{proof}

\end{document}